\documentclass[12pt]{article}

\usepackage[T1]{fontenc}

\usepackage[
left=0.6in,
right=0.6in,
top=0.6in,
bottom=0.6in
]{geometry}

\usepackage{setspace}
\usepackage{amsmath,amssymb,amsfonts}
\usepackage{amsthm}
\newtheorem{remark}{Remark}

\usepackage{algorithm}
\usepackage{algpseudocode}
\usepackage{float}
\usepackage{placeins}
\usepackage{parskip}
\usepackage{mathrsfs}
\usepackage{graphicx}
\usepackage{xcolor}
\usepackage{subcaption}
\usepackage{forest}
\usepackage{hyperref}
\usepackage{booktabs}
\usepackage{longtable}
\usepackage{array}
\usepackage{tabularx}

\usepackage{amsmath,amssymb,amsfonts}

\usepackage{parskip}

\usepackage{threeparttable}

\usepackage{graphicx}
\usepackage{subcaption}
\usepackage{tikz}
\usetikzlibrary{positioning,arrows.meta}
\usepackage{capt-of}
\usepackage{afterpage}
\usepackage{algorithm}
\usepackage{algpseudocode}

\usepackage{xspace}
\usepackage{color}
\usepackage{enumerate}
\usepackage{tikz}
\usetikzlibrary{arrows.meta,positioning}
\usepackage{hyperref}
\usepackage{array}
\usepackage{changepage}
\usepackage{booktabs}   % for \toprule, \midrule, \bottomrule
\algnewcommand{\algorithmicswitch}{\textbf{switch}}
\algnewcommand{\algorithmiccase}{\textbf{case}}

\algdef{SE}[SWITCH]{Switch}{EndSwitch}[1]
{\algorithmicswitch\ #1\ \textbf{do}}
{\textbf{end switch}}

\algdef{SE}[CASE]{Case}{EndCase}[1]
{\algorithmiccase\ #1\textbf{:}}
{\textbf{end case}}

\newenvironment{keywords}
{\par\noindent\textbf{Keywords: }\ignorespaces}
{\par}

\newcommand{\sep}{;\ }

\newtheorem{lemma}{{\bf Lemma}}
\newtheorem{property}{{\bf Property}}
\newtheorem{observation}{{\bf Observation}}
\newtheorem{definition}{{\bf Definition}}
\newtheorem{theorem}{{\bf Theorem}}

 \newcommand{\1}{$ {\mathcal Multiple}$}
\newcommand{\2}{$ {\mathcal Dense}$}
 \newcommand{\3}{$ {\mathcal Non-Dense}$}
  \newcommand{\4}{$ {\mathcal Final}$}

\title{\textbf{Conflicting Pattern Formation by Teams of Anonymous, Fully Disoriented Robots}}

\author{
Animesh Maiti\thanks{Corresponding author: \texttt{p22ma201@iitj.ac.in}}
\and
Prakhar Shukla
\and
Subhash Bhagat
}

\date{
\small
Department of Mathematics, Indian Institute of Technology Jodhpur\\
Jodhpur, Rajasthan, India
}

\begin{document}

\maketitle

\begin{abstract}
Two groups of autonomous, anonymous, and oblivious mobile robots are deployed in the two-dimensional Euclidean plane, each assigned a distinct task. We study a setting where the two groups must simultaneously solve two conflicting pattern formation problems: the \textit{gathering problem}, where robots gather at a point not known to them a priori, and the \textit{circle formation problem}, where robots occupy distinct positions on the boundary of a circle. Although each robot knows its own task, it cannot identify other members of its group. A prior solution~\cite{Conflict-1} addressed this problem for asynchronous robots having {\it direction-only axis agreement} and {\it global weak multiplicity detection} capability available to all robots in both groups. In contrast, in this work, we consider fully {\it disoriented robots} without any axis agreement or common \textit{chirality}. We study the feasibility of a solution to this problem for {\it disoriented robots}. We propose a distributed algorithm that solves the problem for semi-synchronous disoriented robots with non-rigid movements. Our proposed algorithm assumes global weak multiplicity detection only for the gathering group, while for the circle formation group, it requires local weak multiplicity detection. 
\end{abstract}

              % typeset the header of the contribution
\begin{keywords}
\setlength{\parskip}{0pt}
Mobile robots
\sep Conflicting pattern formation
\sep Gathering
\sep Circle formation
\sep Fully disoriented robots
\sep Multiplicity detection
\end{keywords}
\section{Introduction}
\label{sec:introduction}

A {\it swarm of robots} consists of several small autonomous mobile robots that cooperatively perform tasks without any centralized control. The robots are anonymous, i.e., they can not be distinguished by their physical appearances or by identity. Robots are homogeneous, i.e., all the robots have the same set of capabilities, and they run the same distributed algorithm. They are oblivious, i.e., they do not carry forward any information from the previous computational cycles. They do not have explicit communication capabilities. They communicate implicitly by changing their positions. Robots do not have access to any global coordinate system. However, each robot has its own local coordinate system centered at its current position. The direction and the orientation of the local coordinate axes may vary among the robots. 

At a point in time, a robot is either idle or active. The activation of the robots is controlled by an adversarial scheduler. There are mainly three types of scheduler: {\it fully-synchronous} ($\mathrm{FSYNC}$), {\it semi-synchronous} ($\mathrm{SSYNC}$), or {\it asynchronous} ($\mathrm{ASYNC}$). In $\mathrm{FSYNC}$ and $\mathrm{SSYNC}$ scheduling, time is divided into discrete rounds. An $\mathrm{FSYNC}$ scheduler activates all the robots in each round, whereas an $\mathrm{SSYNC}$ scheduler activates only a subset of robots in each round. The $\mathrm{ASYNC}$ scheduler is the most general one. Under this scheduler, there is no notion of common rounds. The activations and execution times are unpredictable but finite. Each active robot repeatedly executes an atomic {\it Look-Compute-Move} cycle. During the {\it Look} phase, a robot observes the positions of other robots with respect to its own local coordinate system. In the {\it Compute} phase, it computes a destination point using the locations obtained in the {\it Look} phase, and in the {\it Move} phase, it moves toward that destination point. We assume a {\it fair scheduler} that activates each robot infinitely many times.

 Robot movements may be {\it rigid} or {\it non-rigid}. For non-rigid  movement, an adversary may stop a robot before it reaches its destination; however, the robot always moves at least $\delta>0$ distance unless it reaches the destination. Robots may be endowed with some extra capabilities, and they may have some agreements on the directions and orientations of the local coordinate axes. Two or more robots may occupy the same point in the plane, and such a point is called a {\it multiplicity point}. If robots have {\it global weak multiplicity detection} capability, then they can identify a multiplicity point even if they do not lie on the multiplicity point. On the other hand, if robots have {\it local weak multiplicity detection} capability, then they can identify a multiplicity point only when they lie at that multiplicity point.

The literature contains a large volume of studies on different geometric pattern formation problems under different computational models~\cite{Santoro2012}. The {\it gathering problem} and the {\it circle formation problem} are two of the fundamental formation problems studied in the literature. The gathering problem asks a set of autonomous robots to coordinate their movements in such a way that, within a finite time, all of them meet at a point not known to them a priori. Whereas the circle formation problem requires the robots to place themselves at distinct points on the boundary of a circle. These two problems have been studied separately under different computational models, and several results and algorithms have been proposed for these two problems in the literature.

\section{Related Work}

% One of the major focuses of research in this field is to establish minimal sets of robot capabilities needed to solve a given geometric formation problem under a given model. 
One of the major research directions in distributed mobile robotics is to determine the minimum set of robot capabilities required to solve geometric pattern formation problems under different computational models. Among these problems, gathering and circle formation are two of the most fundamental and extensively studied tasks in the literature.

{\bf The gathering problem:}
The gathering problem has been studied extensively in the literature under different computational models~\cite{Santoro2012}. The problem is solvable for fully synchronous robots without any extra assumption on the capabilities of the robots~\cite{Santoro2012}; however, the problem is not solvable for $n\ge 2$ semi-synchronous (and hence for asynchronous) robots without additional capabilities of the robots~\cite{Prencipe2007,Suzuki1999}. The problem is solvable in finite time for $n\ge 5$ asynchronous oblivious robots when robots are endowed with global weak multiplicity detection capability~\cite{Cieliebak2012}. The algorithm proposed in~\cite{Viglietta2013} solves the gathering problem for $n=2$ robots using lights with two colors. Under the limited visibility model, the problem is solvable for asynchronous robots with a consistent compass~\cite{Flocchini2005}. The gathering problem has also been studied for $fat$ robots (robots having physical extents). Czyzowicz et al. solved the gathering problem for $n = 4$ fat robots~\cite{czyzowicz2009gathering}. Later, Agathangelou et al. proposed a gathering algorithm for an arbitrary number of robots~\cite{Agathangelou2012}. Their algorithm assumes common chirality. The gathering problem has also been studied under various fault models~\cite{agmon2006fault,BhagatM17,PattanayakMRM19,Bhagat201650,Bouzid2013}. 

{\bf The circle formation problem:}
This problem is solvable for asynchronous robots in finite time without
any additional capabilities of the robots~\cite{Santoro2012}.
The circle formation problem was also studied with the additional
objective of minimizing the maximum distance traversed by any
robot~\cite{bhagat2018optimum}. Several variants of circle formation
have subsequently been investigated. In particular, the $k$-circle
formation problem has been studied for asynchronous robots under
one-axis agreement and, later, for fully disoriented robots
~\cite{bhagat2021kcircle,das2022kcircle}.

One important variation is the {\it uniform circle formation problem},
which requires the robots to occupy distinct equally spaced positions
on a common circle. Uniform circle formation is solvable for
asynchronous robots without additional assumptions
~\cite{flocci2014,vigl2016}. More recent works have considered uniform
circle formation for opaque luminous robots under different
schedulers~\cite{feletti2023uniform}, improved the asynchronous time
complexity to $O(\log n)$ while using a constant number of colors
~\cite{feletti2024log}, and obtained an asymptotically optimal
constant-time, constant-color solution for asynchronous luminous
robots~\cite{feletti2024optimal}. Recently, uniform $k$-circle
formation has also been investigated for asynchronous fat robots
~\cite{das2025uniform}.

{\bf The conflicting patterns formation problem:} The majority of the existing works assume that all robots cooperate to achieve a single common objective. Bhagat et al.~\cite{Conflict-1} initiated the study of conflicting task formation, where two groups of robots simultaneously solve two distinct pattern formation problems. They considered the gathering and circle formation as the conflicting patterns and proposed an algorithm to solve the problem for asynchronous robots. Their algorithm assumes direction-only axis agreement, global weak multiplicity detection capability for both the groups and more than $5$ robots in the gathering group. 

In this paper, we extend this line of work for {\it disoriented robots}, i.e., for robots without any form of axis agreement. Our objective is to investigate the feasibility of a solution for this problem for disoriented robots.

 %%%%%%%%%%%%%%%%%%%%%%%%%%%
\section{Model and Terminologies}
\subsection{Model}
\label{mod}

Let $\mathcal R =\{r_1, r_2, \dots, r_n\}$ denote the set of $n$ robots represented by points in the Euclidean plane. Robots are anonymous and oblivious, and they do not have any form of direct communication capabilities. Robots do not have any form of axis agreement, and they have non-rigid movements. Robots work under the semi-synchronous scheduler (the $\mathrm{SSYNC}$ model). Let $\mathcal R_g\subset \mathcal R$ denote the set of robots required to solve the {\it gathering problem}, and $\mathcal R_f\subset \mathcal R$ be the team of robots required to solve the {\it circle formation problem}. We assume $\mathcal R=\mathcal R_g\cup\mathcal R_f$ and   $|\mathcal R_g| \ge 6$, $|\mathcal R_f|\ge2$. Let $N_f$ and $N_g$ denote the cardinalities of $\mathcal R_f$ and, $\mathcal R_g$ respectively. We assume the following: (i) robots in $\mathcal R_g$ have global weak multiplicity detection capability, and (ii) robots in $\mathcal R_f$ have local weak multiplicity capability, and they know the value of $|\mathcal R_f|$, i.e., $N_f$. A robot $r_i\in\mathcal R$ knows its task to perform, i.e., it knows the team to which it belongs. However, it can not identify its other team members. Initially, all robots are stationary and lie on distinct points on the plane.

\subsection{Terminologies}

Let $r_i(t)$ denote the position of robot $r_i\in\mathcal R$ at time $t$.
Let $\widetilde{\mathcal R}(t)=\{r_1(t),\ldots,r_n(t)\}$ be the multiset
of robot positions occupied by the robots in $\mathcal R$ at time $t$.
Let $\mathcal R(t)
=
\{p\in\mathbb{R}^2 \mid p \text{ occurs in }
\widetilde{\mathcal R}(t)\}$ denote the set of distinct occupied positions of the robots in
$\mathcal R$ at time $t$.
Similarly, let $\mathcal R_g(t)$ and $\mathcal R_f(t)$ denote the sets
of distinct positions occupied by the robots in $\mathcal R_g$ and
$\mathcal R_f$, respectively, at time $t$.
A point $p\in\mathcal R(t)$ occupied by at least two robots is called
a \emph{multiplicity point}. A multiplicity point occupied by at least
two robots from $\mathcal R_g$ is called a \emph{stable multiplicity point}.

Let $S(t)$ denote the smallest enclosing circle ($\operatorname{SEC}$) of the robot positions at time $t$, and let $\mathcal O(t)$ denote its centre. We use $\partial S(t)$ to denote the circumference of $S(t)$. Let $S_{out}(t)$ and $S_{in}(t)$ denote the robot positions in $\mathcal R(t)$ lying on the circumference of $S(t)$ and inside $S(t)$, respectively. $S(0)$ is the smallest enclosing circle for the initial robot configuration $\mathcal R(0)$. Let $D(t)=
\left\{
\|p-\mathcal O(t)\| :
p\in\mathcal R(t),\;
p\neq \mathcal O(t)
\right\}.$ Let $0<\rho_1(t)<\rho_2(t)<\cdots<\rho_m(t)$ be the distinct values in $D(t)$. For each $k\in\{1,\ldots,m\}$,
let $C_k(t)$ denote the circle centered at $\mathcal O(t)$ with
radius $\rho_k(t)$. Thus, the circles are indexed from the innermost
to the outermost radial level, and $C_m(t)=S(t)$. (see Figure~\ref{fig:c_k}). Let $a$ and $b$ be two points in the plane. By $(a,b)$ and $\overline{ab}$, we denote the open (excluding $a$ and $b$) and closed (including $a$ and $b$) line segments joining $a$ and $b$, respectively.

\begin{figure}[htbp]
    \centering
    \includegraphics[width=0.5\textwidth]{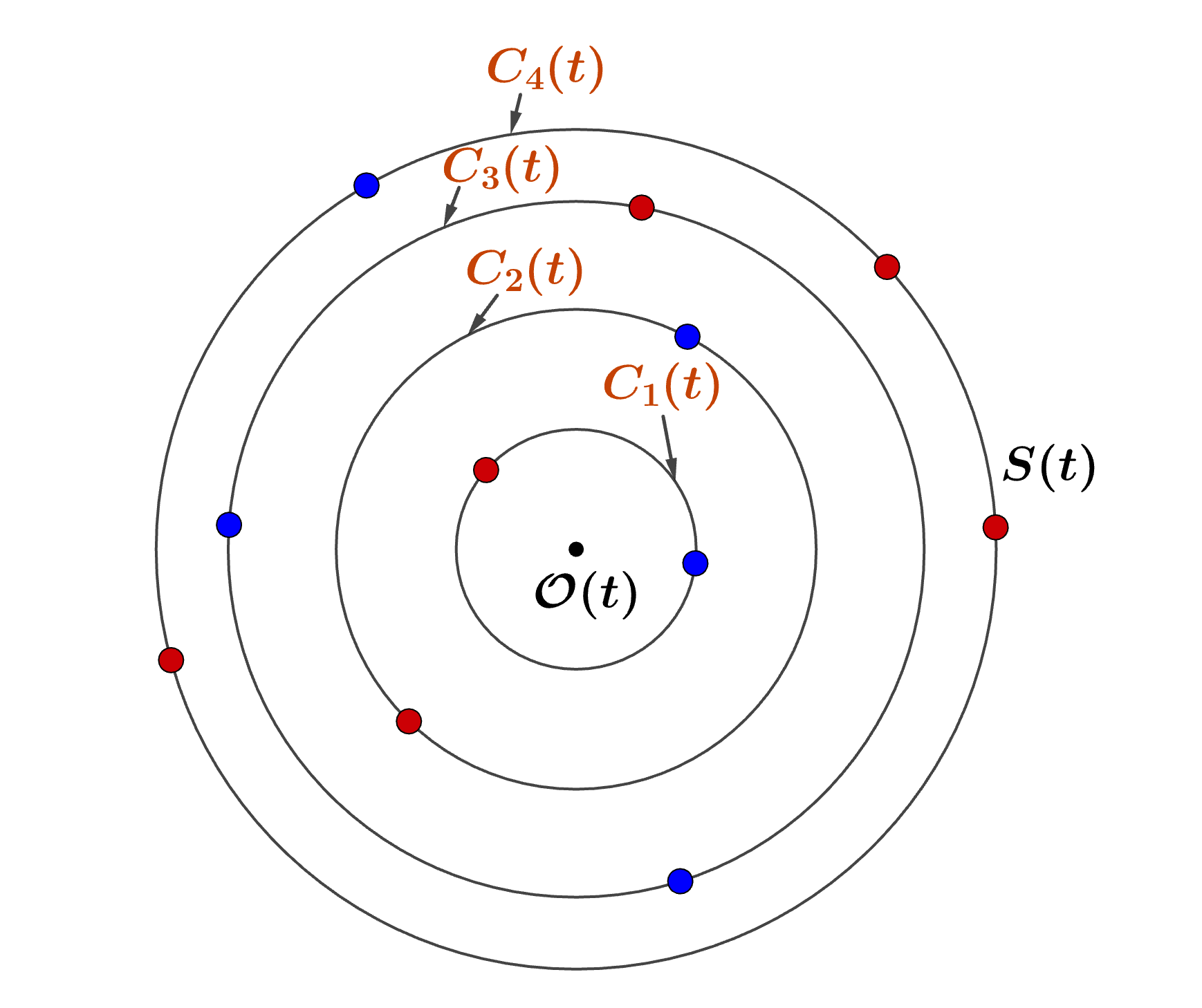}

    \caption{Illustrations of the circles $C_k(t)$ for $k \geq 1$: blue robots represent the set $\mathcal{R}_f$ and red robots represent the set $\mathcal{R}_g$.}

\label{fig:c_k}
\label{App-f11}
\end{figure}

We use the concept of \emph{quasi-regularity} originally studied in~\cite{abs-1207-0226,BhagatM17}. Before defining quasi-regularity, we introduce a few related definitions.

\begin{definition}[Successor]
Let $\mathcal R$ be a set of robot positions in $\mathbb{R}^2$ and let $c \in \mathbb{R}^2$ be a point. For any robot position $r_i \in \mathcal R$, the \emph{successor} of $r_i$ with respect to $c$, denoted by $S(r_i,c)$, is the next robot position in $\mathcal R$ encountered when moving in clockwise order around $c$. Robot positions lying on the same ray from $c$ are ordered consecutively according to their increasing distance from $c$. The $k$-th successor is defined recursively as $S^k(r_i,c)=S(S^{k-1}(r_i,c),c),$ with $S^0(r_i,c)=r_i$. (see Figure~\ref{Three}(b))
\end{definition}

\begin{definition}[String of Angles]
Let $\mathcal R$ be a set of robot positions and let $c \in \mathbb{R}^2$. Let $r_1,r_2,\dots,r_n$ be the robot positions of $\mathcal R$ ordered clockwise around $c$ according to the successor relation. The \emph{string of angles} of $\mathcal R$ with respect to $c$, denoted by $SA(\mathcal R,c)$, is defined as $SA(\mathcal R,c)=(\alpha_1,\alpha_2,\dots,\alpha_n),$ where $\alpha_i=\angle(r_i,c,r_{i+1})
\ \text{for } i=1,\dots,n,$ and $r_{n+1}=r_1.$
Note that if multiple robots lie on the same ray from $c$, then they appear consecutively in the ordering and contribute zero angular gaps in the string of angles. Thus, $SA(\mathcal R,c)$ captures both the angular ordering and the multiplicities of robots. (see Figure~\ref{Three}(b) )
\end{definition}

\begin{figure}[htbp]
    \centering

    \begin{subfigure}[t]{0.40\textwidth}
        \centering
        \includegraphics[width=\textwidth]{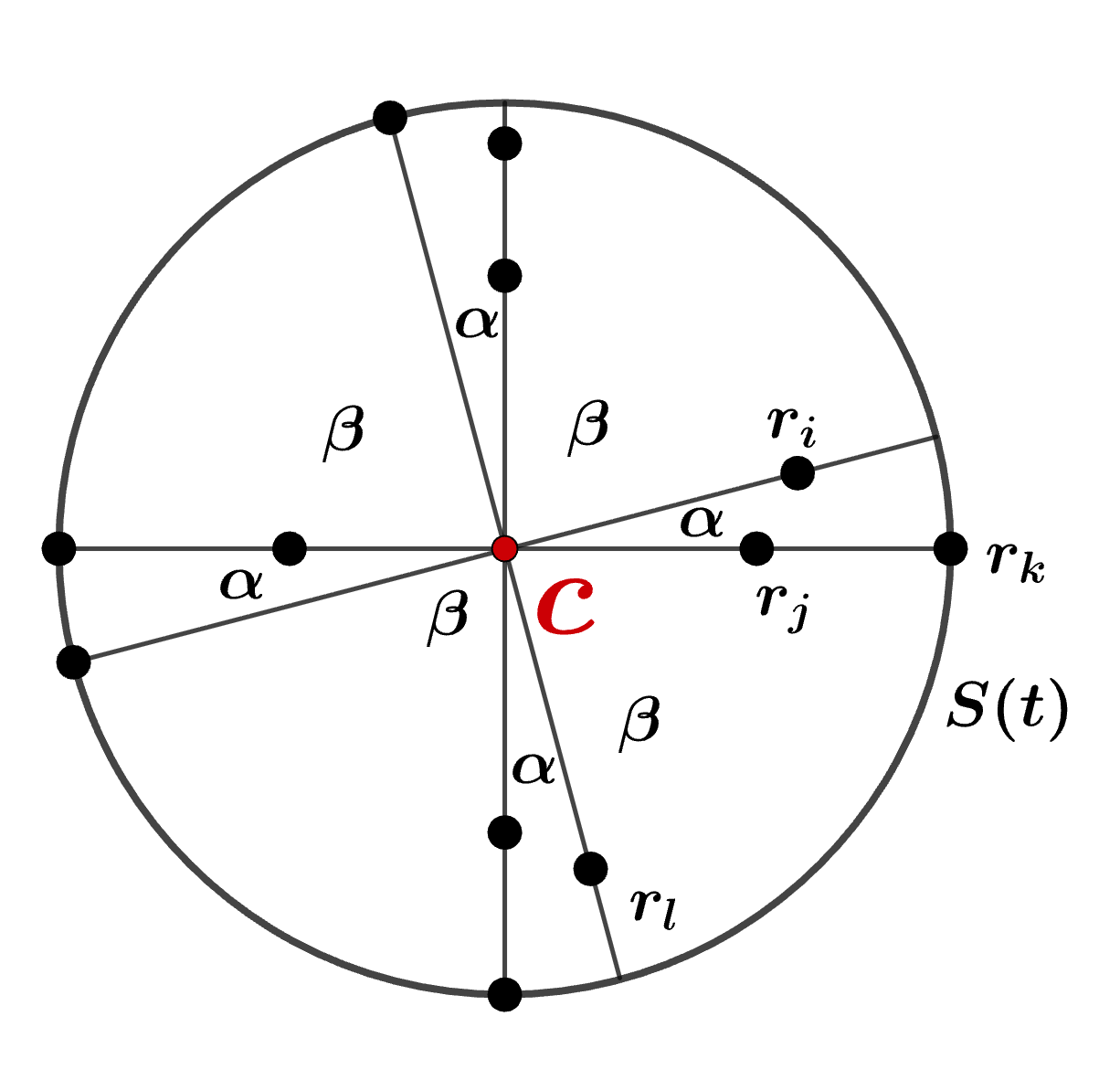}
        \caption{}
        
    \end{subfigure}
   \hspace{0.04\textwidth}%
    \begin{subfigure}[t]{0.40\textwidth}
        \centering
        \includegraphics[width=1.2\textwidth]{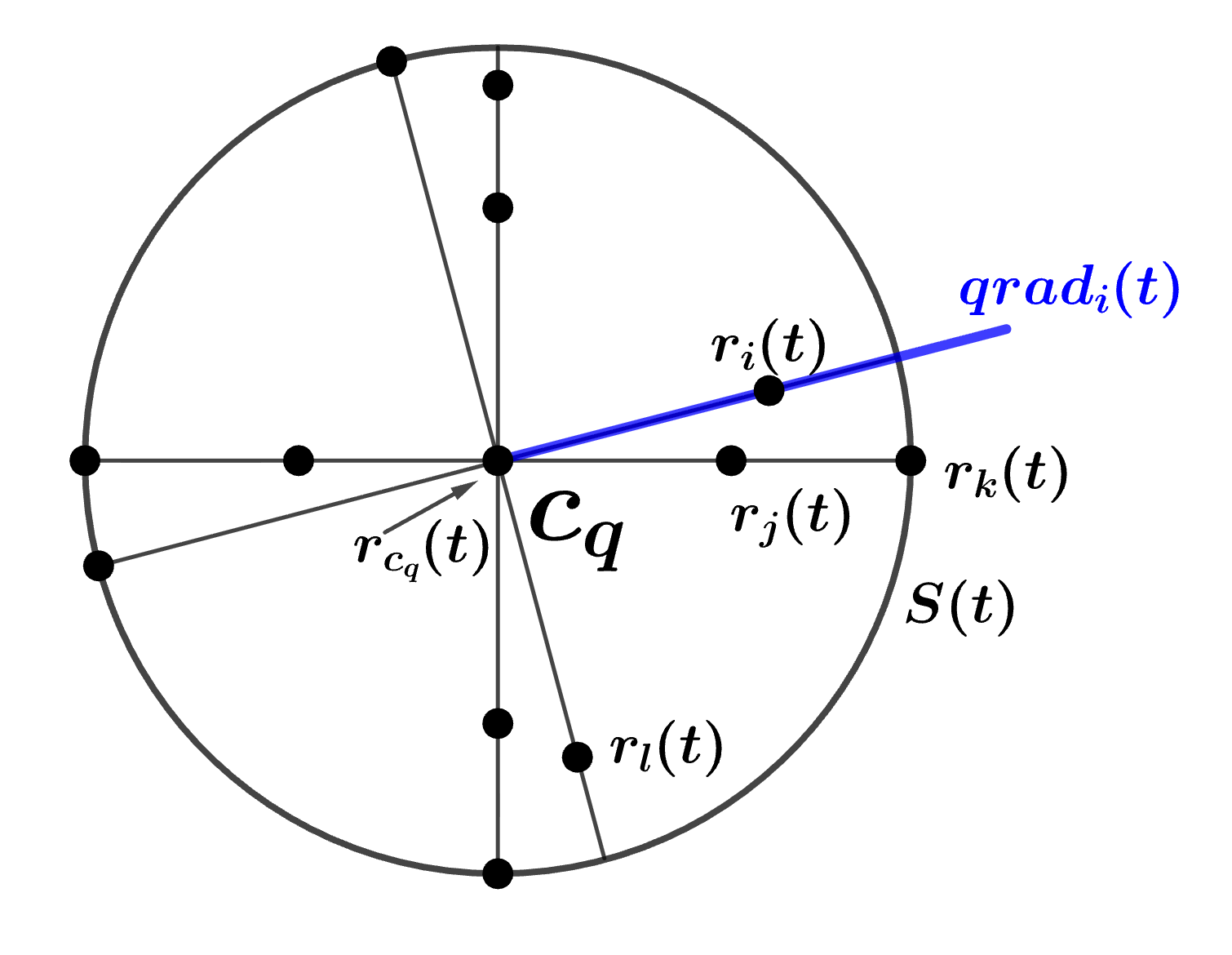}
        \caption{}
        
    \end{subfigure}

   \caption{(a) Illustration of the successor relation and the string of angles around the centre $c$. In the clockwise ordering, $S(r_i,c)=r_j$, $S(r_j,c)=r_k$, $S(r_k,c)=r_l$, and so on. The string of angles is $SA(\mathcal R,c)=(\alpha,0,\beta,\alpha,0,\beta,\alpha,0,\beta,\alpha,0,\beta)$, which can be written as $SA(\mathcal R,c)=X^4,
\quad \text{where } X=(\alpha,0,\beta)$. Therefore, the configuration is regular with $reg(\mathcal R)=4$. (b) Illustration of a quasi-regular configuration $\mathcal R(t)$ with centre of quasi-regularity $c_q$, where $\mathcal B(t)=\mathcal R(t)\setminus \{r_{c_q}(t)\}$. The half-line $qrad_i(t)$ denotes the ray starting from $c_q$ and passing through $r_i(t)$.}

\label{Three}
\end{figure}

\begin{definition}[Regular Configuration]
A configuration $\mathcal R$ is said to be \emph{regular} with respect to a point $c$ if the string of angles $SA(\mathcal R,c)$ can be written as $SA(\mathcal R,c)=X^k,$ for some non-empty sequence $X$ and integer $k>1$. The regularity of $\mathcal R$, denoted by $reg(\mathcal R)$, is defined as the maximum such integer $k$. The point $c$ is called the \emph{centre of regularity}. The point $c$ is called the \emph{centre of regularity}. (see Figure~\ref{Three}(b))
\end{definition}

\begin{definition}[Quasi-Regularity]
\label{def-Q}
Let $\mathcal R(t)$ be a configuration of robots at time $t$. The configuration $\mathcal R(t)$ is said to be \emph{quasi-regular (Q-regular)} if there exists a subset $\mathcal B(t)\subseteq \mathcal R(t)$ and a point $c\in \mathbb R^2$ such that the subset $\mathcal B(t)$ forms a regular configuration with respect to $c$, i.e., the string of angles $SA(\mathcal B(t),c)$ is periodic with period greater than $1$, and all robots in $\mathcal R(t)\setminus \mathcal B(t)$ lie at the point $c$. The point $c$ is called the \emph{centre of quasi-regularity}, denoted by $c_q$. The quasi-regularity of $\mathcal R(t)$ is defined as $qreg(\mathcal R(t))=reg(\mathcal B(t))$, and $qreg(\mathcal R(t))=1$ if no such subset exists. For each robot $r_i(t)\in \mathcal R(t)$ with $r_i(t)\neq c_q$, the half-line $qrad_i(t)$ is defined as the ray starting from $c_q$ and passing through $r_i(t)$. (see Figure~\ref{Three}(b))
\end{definition}

Let $\mathcal{R}(t)$ be a quasi-regular configuration with center $c_q$. For each robot $r_i(t)$, let $qrad_i(t)$ be the segment joining $c_q$ to $r_i(t)$. The configuration is said to be \emph{free-path quasi-regular} if no other robot lies on $qrad_i(t)$ between $c_q$ and $r_i(t)$ for any $r_i(t) \in \mathcal{R}(t)$ (see Figure~\ref{QR}(a)). The configuration is said to be \emph{totally symmetric} if $\mathcal{O}(t) = c_q$ (see Figure~\ref{QR}(a)); otherwise, $\mathcal{O}(t) \neq c_q$ (see Figure~\ref{QR}(c)). 

\begin{definition}[Radial ray]
\label{def:rad}
For a robot position $r_i(t)\in\mathcal R(t)$ with $r_i(t)\ne\mathcal O(t)$,
the half-line $rad_i(t)$ is defined as the ray starting from
$\mathcal O(t)$ and passing through $r_i(t)$. (Compare $qrad_i(t)$ in
Definition~\ref{def-Q}, which is defined analogously with respect to the centre
of quasi-regularity $c_q$ rather than $\mathcal O(t)$.)
\end{definition}

\begin{figure}[htbp]
    \centering

    \begin{subfigure}[t]{0.3\textwidth}
        \centering
        \includegraphics[width=1.25\textwidth]{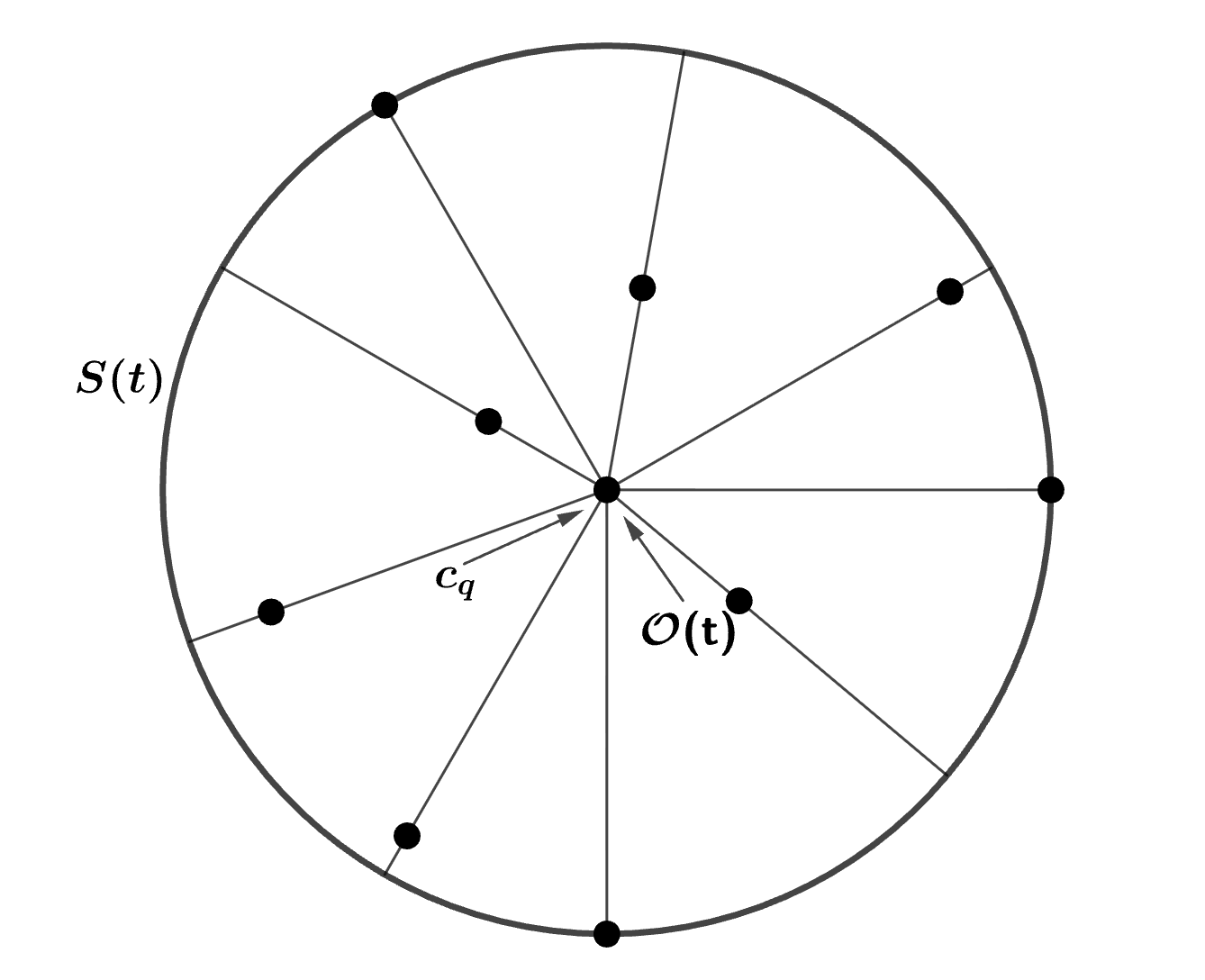}
        \caption{}

    \end{subfigure}
    \hspace{0.03\textwidth}
    \begin{subfigure}[t]{0.3\textwidth}
        \centering
        \includegraphics[width=1.2\textwidth]{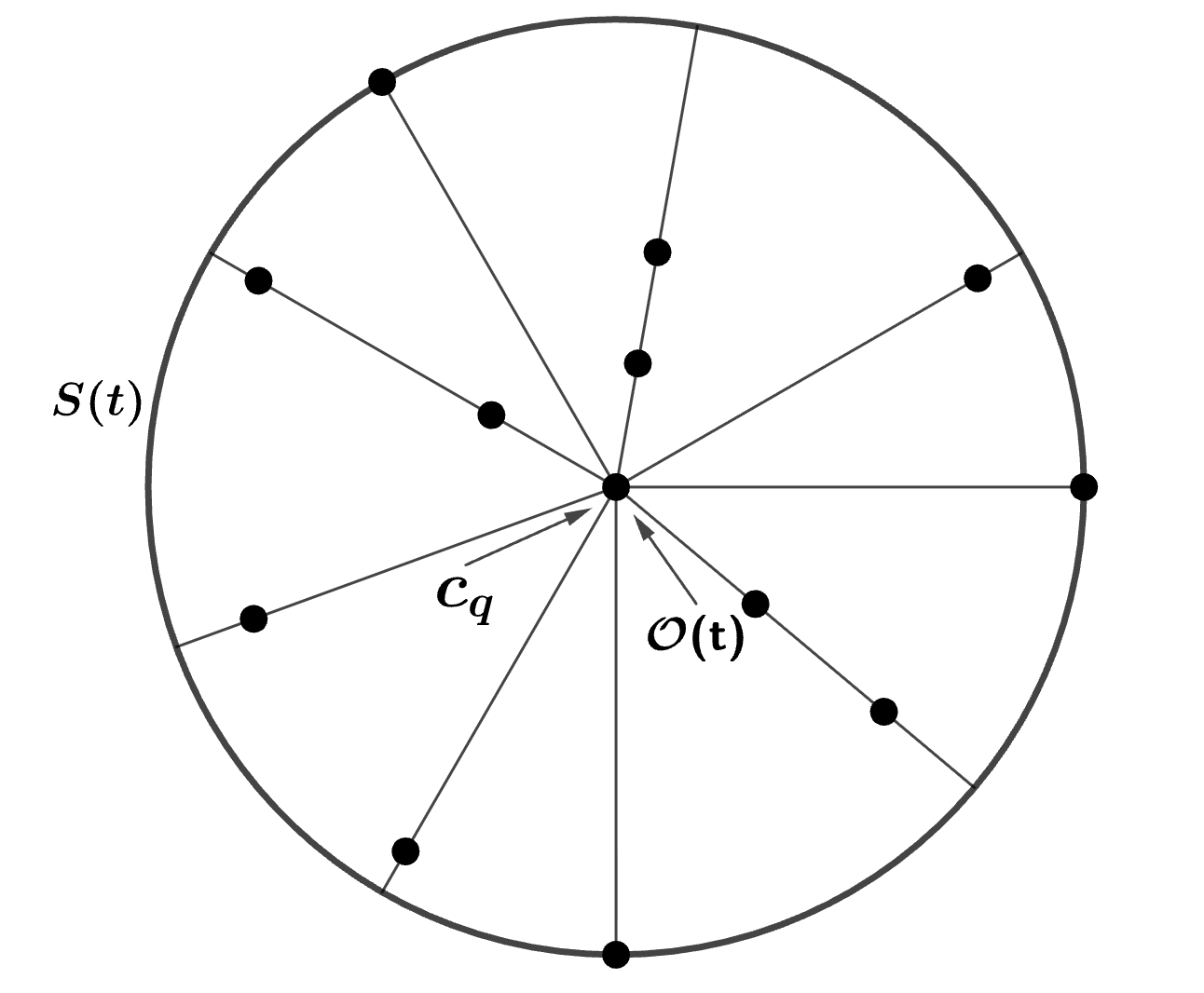}
        \caption{}
  
    \end{subfigure}
    \hspace{0.03\textwidth}
    \begin{subfigure}[t]{0.3\textwidth}
        \centering
        \includegraphics[width=1.2\textwidth]{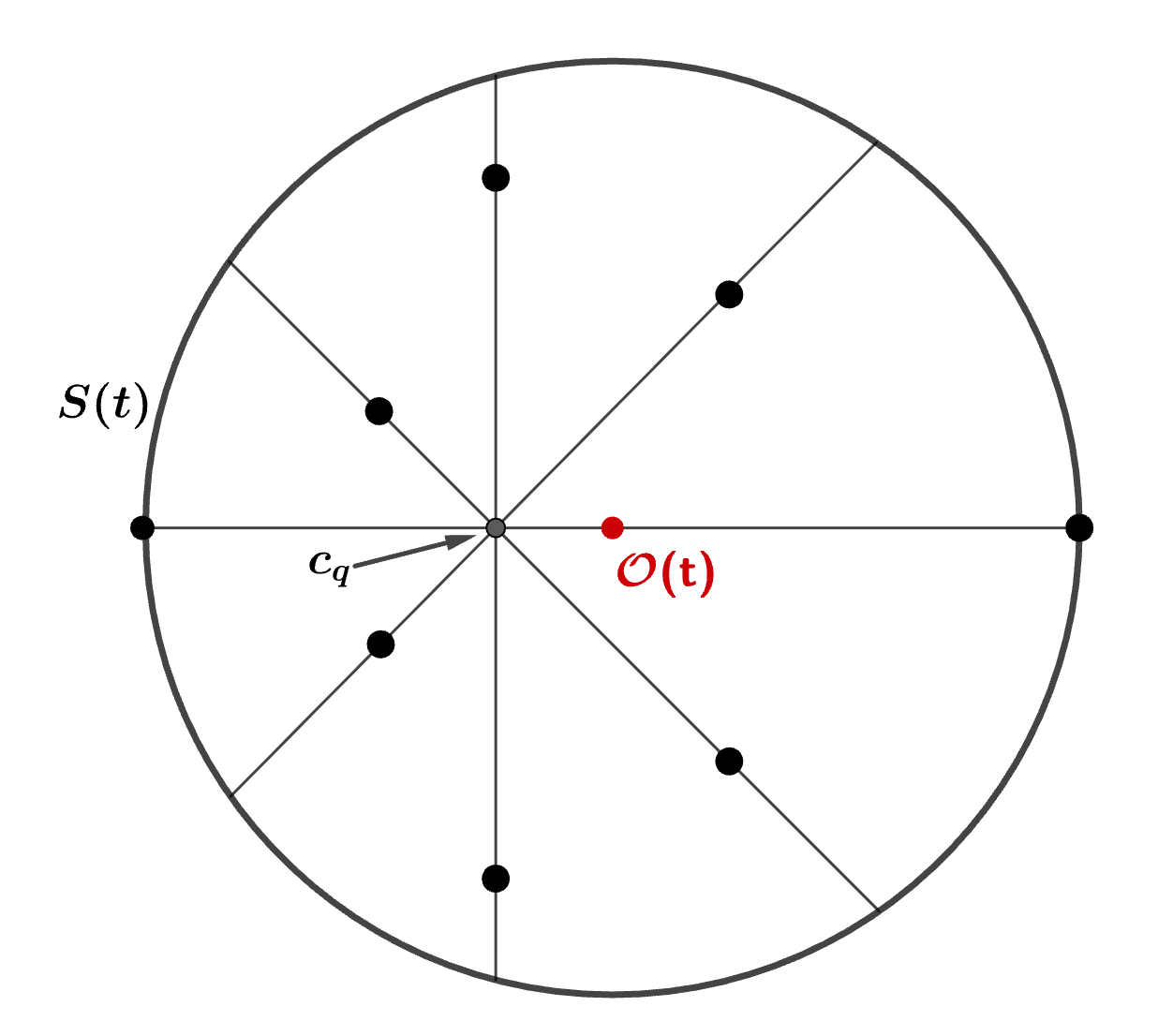}
        \caption{}
  
    \end{subfigure}

\caption{Illustration of different types of quasi-regular configurations:
(a) a \textit{free-path quasi-regular} robot configuration which is also
\textit{totally symmetric}, (b) a \textit{quasi-regular configuration}
which is not \textit{free-path quasi-regular}, (c) a \textit{quasi-regular}
configuration with $\mathcal O(t)\neq c_q$.}
\label{QR}
\end{figure}

To describe our algorithm, we will consider the following  classes of configurations:

\textbf{\boldmath \1:} This class contains all the robot configurations having exactly one multiplicity point.
 
\textbf{\boldmath \2:} This class contains all robot configurations having no multiplicity points, with  $|S_{out}(t)|\le 4$. 

\textbf{\boldmath \3:} A robot configuration belongs to this class if it does not contain any multiplicity points, with $|S_{out}(t)|>4$.

\textbf{\boldmath \4:} A robot configuration $\mathcal{R}(t)$ belongs to this class if it satisfies both of the following conditions:

\begin{enumerate}[(i)]
    \item $\exists\, p^* \in \mathbb{R}^2$ such that
$ r_i(t') = p^*, \quad \forall\, r_i \in \mathcal R_g,\; \forall\, t' \geq t,$  i.e., all the robots in $\mathcal R_g$ achieve gathering.

    \item $\exists$ a circle $S \subset \mathbb{R}^2$ such that $ r_j(t') \in \partial S, \quad \forall\, r_j \in \mathcal R_f,\; \forall\, t' \geq t,$  i.e., all the robots in $\mathcal R_f$ achieve circle formation.
\end{enumerate}

%%%%%%%%%%%%%%%%%%%%%%%%%%%  
\section{Algorithm \textsc{PatternFormation()}}

  The algorithm proposed in~\cite{Conflict-1} has three phases: the {\it reduction phase}, the {\it multiplicity phase} and the {\it formation phase}. The reduction phase ensures that at least two gathering robots are inside the $\operatorname{SEC}$, $S(0)$. This phase depends on the assumption of axis-only agreements, and this helps to keep $S(0)$ intact. Keeping $S(0)$ intact ensures finite-time termination of the reduction phase with at least two robots inside $S(0)$ from the gathering group $\mathcal R_g$. Furthermore, this also helps to decide and keep the final gathering point (the centre of $S(0)$) invariant during the execution of the whole algorithm. The global weak multiplicity detection for the robots forming a circle helps to design synchronized movements of the robots during the formation phase, which is essential to avoid collisions and the creation of multiple multiplicity points (robots use multiplicity detection capability to identify the gathering point).

\subsection{Challenges}

To solve the gathering problem, the robots must first establish a unique and
persistent gathering point. Our approach achieves this by creating a unique stable multiplicity point, which is subsequently identified by robots in
\(\mathcal{R}_g\) using global weak multiplicity detection. Since the robots are
anonymous and cannot distinguish members of their own team, creating such a
multiplicity point is non-trivial: it must not be formed solely by robots in
\(\mathcal{R}_f\), and if an \(\mathcal{R}_f\) robot participates in its
creation, it must remain there until the multiplicity becomes stable. To this
end, robots in \(\mathcal{R}_f\) are equipped with local weak multiplicity
detection. During the multiplicity-creation phase, the initial smallest
enclosing circle is preserved by fixing a small set of pivotal boundary robots.
The assumption \(|\mathcal{R}_g|\ge6\) guarantees that at least two gathering
robots remain in the interior of the $\operatorname{SEC}$ to create the stable multiplicity
point. Consequently, unlike the algorithm of~~\cite{Conflict-1}, new movement
strategies are required. The main challenges are:
\begin{enumerate}
    \item creating a unique stable multiplicity point while handling symmetric
    configurations without any axis agreement;
    \item ensuring that no additional stable multiplicity point is created
    throughout the execution; and
    \item coordinating the gathering and circle-formation phases so that both
    terminate within a finite time.
\end{enumerate}
%--------------------------
\subsection{Overview of the Algorithm}

Algorithm \textsc{PatternFormation()} consists of three phases:
\emph{multiplicity creation}, \emph{gathering}, and \emph{circle formation}.
These phases are not globally synchronized. In particular, gathering and circle
formation may overlap after a multiplicity point has been created.

If the initial configuration is \3, the algorithm first moves selected
non-pivotal boundary robots into the interior of the current smallest enclosing
circle \(S(t)\). By Lemma~\ref{lemma-sec}, the pivotal robots preserve
\(S(t)\). This reduction either creates a stable multiplicity point directly or produces a \2 configuration.

In a \2 configuration, at most four robot positions lie on
\(\partial S(t)\). Since \(|\mathcal R_g|\geq 6\), at least two gathering robots lie
inside \(S(t)\). These robots move toward a common geometrically defined point,
namely the quasi-regular centre \(c_q\) when the configuration is free-path
quasi-regular, and $\mathcal O(t)$ otherwise. Hence, a unique stable multiplicity point
$p_m$ is eventually created.

Once $p_m$ exists, every gathering robot moves toward $p_m$. At the same
time, circle-forming robots may start moving outward when $|\mathcal R(t)|\leq N_f+1$. The movement rules ensure that gathering remains directed toward $p_m$ and
that circle-forming robots occupy distinct positions on the current $\operatorname{SEC}$. The high-level execution flow of Algorithm~\textsc{PatternFormation()} is illustrated in Figure~\ref{fig:state-transition}.

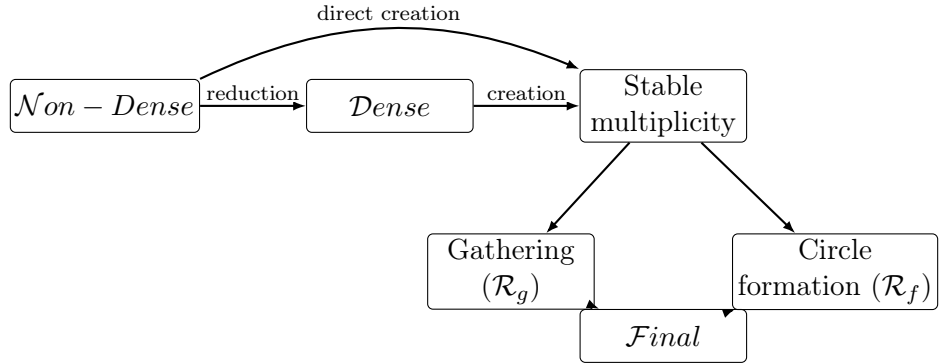
\begin{figure}[htbp]
\centering

\begin{tikzpicture}[
    box/.style={
        draw,
        rounded corners=2pt,
        align=center,
        minimum width=22mm,
        minimum height=7mm,
        inner sep=2pt,
        font=\small
    },
    arr/.style={
        -{Latex[length=1.7mm]},
        thick
    },
    lab/.style={
        font=\scriptsize,
        fill=white,
        inner sep=1pt
    }
]

\node[box] (nd) {\3};
\node[box, right=14mm of nd] (dense) {\2};
\node[box, right=14mm of dense] (mult) {Stable\\multiplicity};

\node[box, below left=12mm and -2mm of mult]
    (gather) {Gathering\\($\mathcal{R}_g$)};

\node[box, below right=12mm and -2mm of mult]
    (form) {Circle\\formation ($\mathcal{R}_f$)};

\node[box, below=22mm of mult] (final) {\4};

\draw[arr]
    (nd) -- node[lab, above, yshift=1pt] {reduction} (dense);

\draw[arr]
    (dense) -- node[lab, above, yshift=1pt] {creation} (mult);

\draw[arr, bend left=24]
    (nd.north east) to
    node[lab, above, yshift=2pt] {direct creation}
    (mult.north west);

\draw[arr] (mult) -- (gather);
\draw[arr] (mult) -- (form);

\draw[arr] (gather) -- (final);
\draw[arr] (form) -- (final);

\end{tikzpicture}

\caption{High-level execution of Algorithm
\textsc{PatternFormation()}.}
\label{fig:state-transition}

\end{figure}

%--------------------------

% In order to implement the above-outlined ideas, the main challenges are faced due to the anonymity of the robots, the semi-synchrony of the scheduler, and the non-rigid movements of the robots. 
The algorithm applies the following priority order:

\begin{enumerate}
    \item if a multiplicity point exists, execute the gathering or formation
          rule according to the number of distinct robot positions in the system.
    \item otherwise, reduce a \3 configuration into \2 configuration while preserving the $\operatorname{SEC}$;
    \item in a \2 configuration, create a stable multiplicity point.
\end{enumerate}

\subsection{Routine Pivotal-Robot-Position-Selection()}

When the current configuration is not totally symmetric, our algorithm computes
a subset $\mathcal P'(t)\subseteq S_{\mathrm{out}}(t),$
called the \emph{pivotal set}. A subset
$\mathcal P'(t)\subseteq S_{\mathrm{out}}(t)$ is called a \emph{pivotal set} if
$\operatorname{SEC}(\mathcal P'(t))=S(t).$  Robots occupying pivotal positions remain
stationary during the multiplicity-creation phase, while the remaining boundary
robots may move. The pivotal set is computed according to the geometric symmetry of the current configuration.
\begin{itemize}
    \item {\it Asymmetric configurations:}
Let \(p_0\) be the first robot position in $S_{out}(t)$ under the canonical ordering of the configuration (since $\mathcal R(t)$ is asymmetric, we can obtain an ordering of the robot positions in $\mathcal R(t)$ (\cite{chaudhuri2015leader}). Let \(\mathcal L\) be the line through \(p_0\) and $\mathcal O(t)$. If the antipodal point of \(p_0\) is occupied by a robot position in $S_{out}(t)$, then \(\mathcal P'(t)\) consists of these two antipodal positions. Otherwise, \(\mathcal P'(t)\) contains \(p_0\) together with the two robot positions in $S_{out}(t)$ adjacent to its antipodal point.

\item {\it Configurations with one symmetry axis:}
Let $\mathcal L$ be the unique symmetry axis, and let \(a,b\in\partial S(t)\) be its
intersections with \(\partial S(t)\). The pivotal set consists of every
occupied intersection and, for each unoccupied intersection \(x\in\{a,b\}\),
the first occupied boundary positions encountered from \(x\) along the two
circular directions of \(\partial S(t)\), after removing duplicates.(see Figure~\ref{pivotes})

\item {\it Quasi-regular configurations with $c_q\neq \mathcal O(t)$.}
 Let \(D\) denote the set of robot positions in $S_{out}(t)$
having maximum distance from \(c_q\). Note that $1\le D\le2$, as $c_q\neq O(t)$. The canonical line is defined as the line
through $\mathcal O(t)$ and the unique point determined by \(D\): if \(D\) consists of
a single position, that position is used; otherwise, the midpoint of the
positions in \(D\) is used. The pivotal set is then computed exactly as in the
single-symmetry-axis case.
\end{itemize}

\begin{property}
\label{prop:sec}
Let $A$ be a set of points in the Euclidean plane with $|A|\ge 3$, and
let $B\subseteq A$ with $2\le |B|\le 4$. Consider the non-overlapping
division of the circumference of the smallest enclosing circle of $A$
into arcs by the points of $B$ lying on that circumference (or, if
$|B|=2$ and the two points are diametrically opposite, the trivial
division by the diameter). If no arc exceeds a semicircle, then the
smallest enclosing circles of $A$ and $B$ coincide, i.e., $B$ alone
suffices to determine the smallest enclosing circle of $A$.
\end{property}

\begin{figure}[htbp]
    \centering

    \begin{subfigure}[t]{0.3\textwidth}
        \centering
        \includegraphics[width=\textwidth]{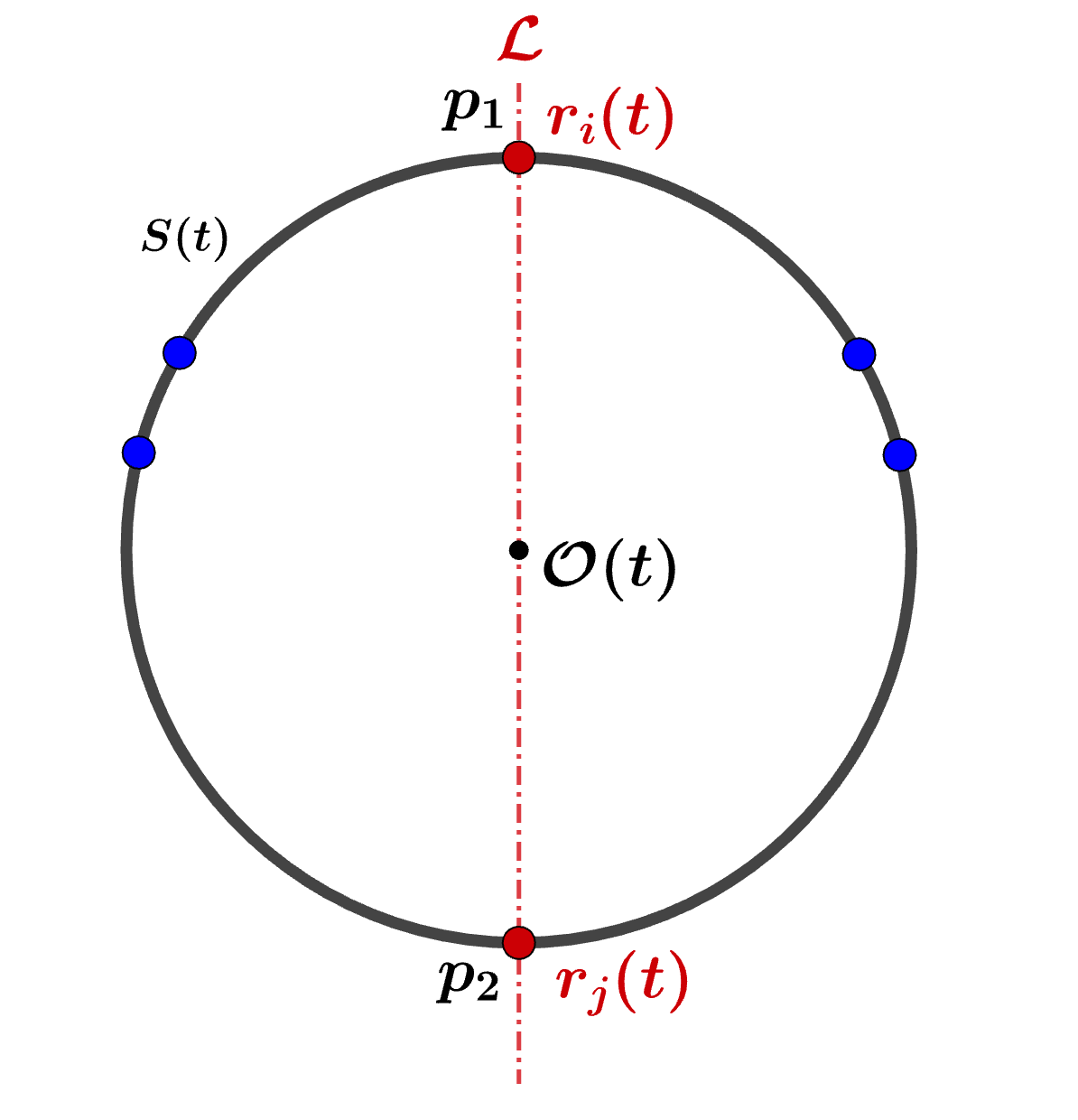}
        \caption{}
    \end{subfigure}
    \hspace{0.03\textwidth}
    \begin{subfigure}[t]{0.31\textwidth}
        \centering
        \hspace*{-0.15\textwidth} % shifts only figure (b) slightly left
        \includegraphics[width=1.2\textwidth]{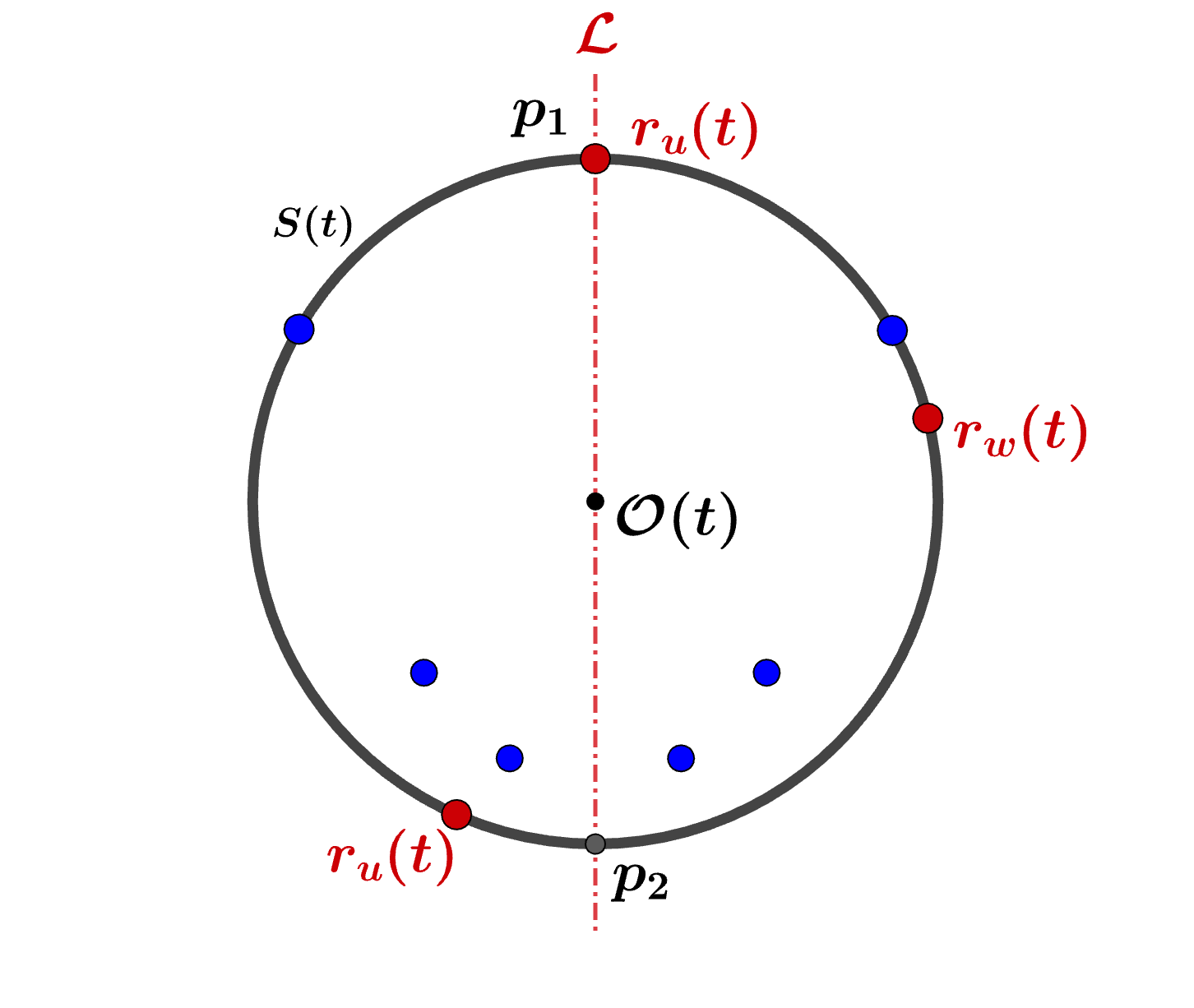}
        \caption{}
    \end{subfigure}
    \hspace{0.03\textwidth}
    \begin{subfigure}[t]{0.3\textwidth}
        \centering
        \includegraphics[width=\textwidth]{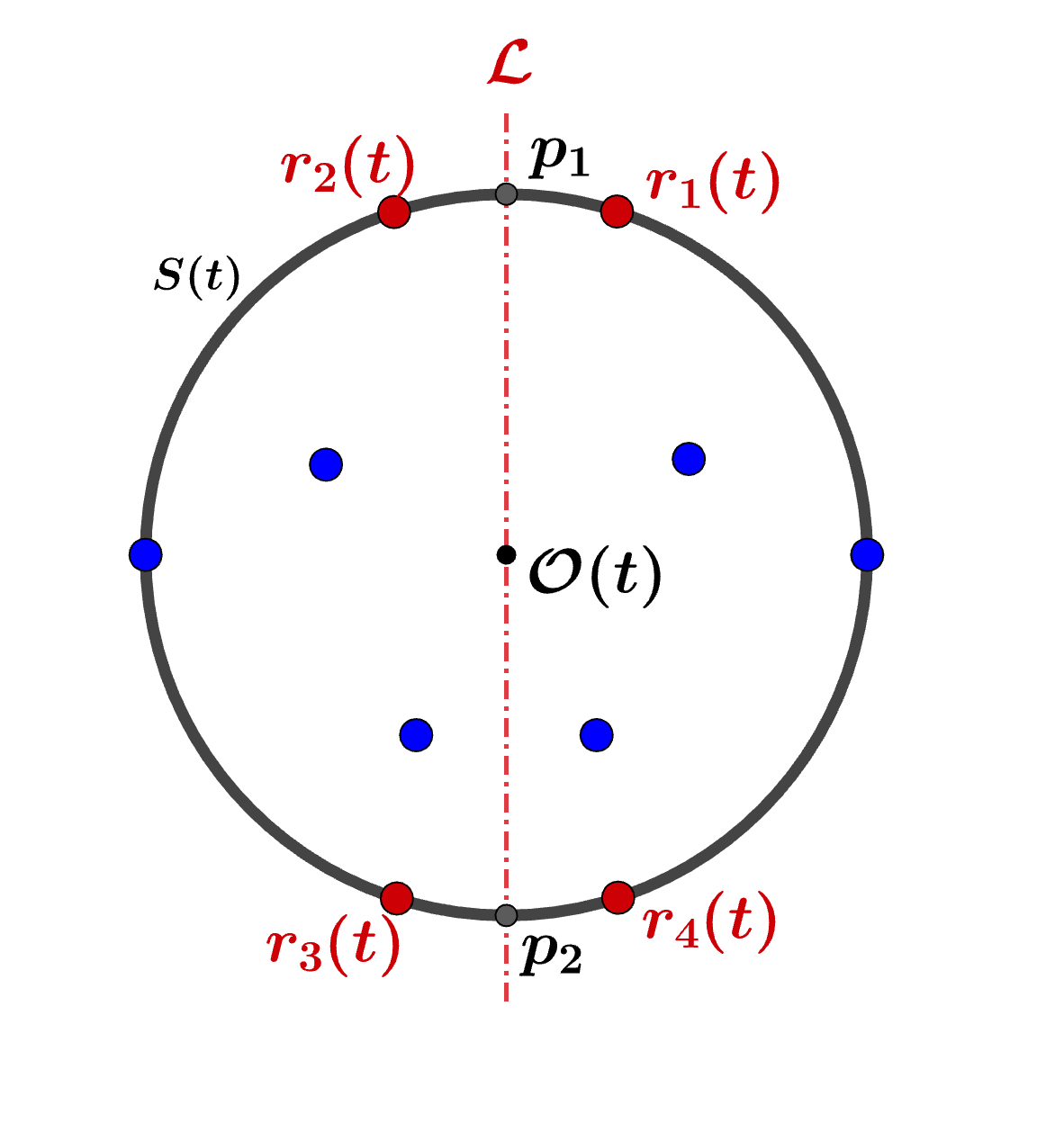}
        \caption{}
    \end{subfigure}

    \caption{Illustrations of different scenarios of pivotal selection:
    (a) both $p_1$ and $p_2$ contain robot positions,
    (b) exactly one of $p_1$ and $p_2$ contains a robot position, and
    (c) none of $p_1$ and $p_2$ contains a robot position.}
    \label{pivotes}
\end{figure}

\begin{lemma}
\label{lemma-sec}
For every configuration in which the pivotal set is defined, the smallest
enclosing circle of $\mathcal P'(t)$ coincides with $S(t)$; that is, $\operatorname{SEC}\bigl(\mathcal P'(t)\bigr)=S(t).$
\end{lemma}

\begin{proof}
By construction, \(\mathcal P'(t)\) contains either (i) two antipodal boundary positions of \(S(t)\), or (ii) at least three occupied boundary positions that are not contained in any open semicircle of \(S(t)\). The property \ref{prop:sec} of the smallest enclosing circle implies that
either condition uniquely determines \(S(t)\). Therefore, $\operatorname{SEC}(\mathcal P'(t)) =
S(t).$ Let $S(t)$ be the smallest enclosing circle ($\operatorname{SEC}$) of the configuration $\mathcal R(t)$ with center $\mathcal O(t)$, and let $\mathcal P'(t) \subseteq S_{out}(t)$ be the set of pivotal robot positions. We prove that $\mathcal P'(t)$ uniquely determines $S(t)$ using Property~\ref{prop:sec}.

\medskip

\noindent
\textbf{Step 1: Characterization of $\operatorname{SEC}$.}
It is well known that the $\operatorname{SEC}$ of a point set is uniquely determined by either:
\begin{itemize}
    \item[(i)] two diametrically opposite boundary points, or
    \item[(ii)] at least three boundary points not contained in any open semicircle.
\end{itemize}

\medskip

\noindent
\textbf{Step 2: Diametrically opposite case.}
If $\mathcal P'(t)$ contains two diametrically opposite points, then these two points uniquely define $S(t)$. Note that this may occur even when $|\mathcal P'(t)| \geq 3$.

\medskip

\noindent
\textbf{Step 3: Non-diametrical case.}
Assume that no two points in $\mathcal P'(t)$ are diametrically opposite. Then $|\mathcal P'(t)| \geq 3$.

Order the points of $\mathcal P'(t)$ along $S(t)$ in clockwise order: $r_{i_1}(t), r_{i_2}(t), \dots, r_{i_k}(t), \ k = |\mathcal P'(t)| \ and \ k=3, 4.$ These points partition the circumference into $k$ arcs. Let $\alpha_j$ denote the arc between $r_{i_j}(t)$ and $r_{i_{j+1}}(t)$ (indices modulo $k$).

\medskip

\noindent
\textbf{Step 4: Arc bound.}
We claim that $\alpha_j \leq \pi R$ for all $j$, i.e., no arc exceeds a semicircle.
Suppose, for contradiction, that there exists an arc $\alpha_m > \pi R$. Then all points of $\mathcal P'(t)$ lie within the complementary semicircle, say $S_c(t)$.
We first show that if all points of $\mathcal P'(t)$ lie within the semicircle,  $S_c(t)$, then all points of $\mathcal R(t)$ must also lie within the semicircle,  $S_c(t)$. Now, by construction, the points in $\mathcal P'(t)$ are obtained using a well-defined line $\mathcal L(t)$ passing through $\mathcal O(t)$. The line $\mathcal L(t)$ divides the circle $S(t)$ into two semicircles, say, $S_1(t)$ and  $S_2(t)$. If $|\mathcal P'(t)|=3$, there is a unique robot position, say $r_i(t)$ on the boundary of $S(t)$ that lies on $\mathcal L(t)$ and the other two robot positions in $\mathcal P'(t)$, say $r_j(t)$ and $r_k(t)$, are farthest robot positions from $r_i(t)$ lying on $S(t)$. When  $|\mathcal P'(t)|=4$, then there are two pairs of points in $\mathcal P'(t)$ such that one pair lies on $S_1$ and the other pair lies on $S_2$. Now consider any one of these pairs, say the pair lying on $S_1(t)$. Then points in this pair lie on the same side of $\mathcal L(t)$, and they are the farthest robot positions on $S_1(t)$. Thus $\mathcal P'(t)$ captures extremal or symmetry-critical positions of $\mathcal R(t)$. Hence, if $\mathcal P'(t)$ lies within a semicircle, then all points of $\mathcal R(t)$ must also lie within that semicircle.
This contradicts a fundamental property of the $\operatorname{SEC}$: no open semicircle of $S(t)$ can contain all points of $\mathcal R(t)$, otherwise a smaller enclosing circle would exist.

\medskip

\noindent
\textbf{Step 5: Conclusion.}
Thus, no arc between consecutive points of $\mathcal P'(t)$ exceeds a semicircle. Therefore, $\mathcal P'(t)$ is not contained in any open semicircle.

By Property~\ref{prop:sec}, $\mathcal P'(t)$ uniquely determines $S(t)$.
\end{proof}

\begin{algorithm}[!htbp]
\caption{: \textsc{MoveToDestination}$(r_i,\tau,d)$}
\label{alg:move-to-destination}
\begin{algorithmic}[1]
\Require Robot $r_i$, movement type
$\tau\in\{\textsc{StepIn},\textsc{StepAside},\textsc{StepOut}\}$,
and reference point $d$ whenever required

\If{$\tau=\textsc{StepIn}$}
    \State \Call{StepIn}{$r_i,d$}
\ElsIf{$\tau=\textsc{StepAside}$}
    \State \Call{StepAside}{$r_i,d$}
\ElsIf{$\tau=\textsc{StepOut}$}
    \State \Call{StepOut}{$r_i$}
\Else
    \State \textbf{remain stationary}
\EndIf
\end{algorithmic}
\end{algorithm}

 %\vspace*{-0.5cm}
\subsection{Routine \textsc{MoveToDestination()}}
%%\vspace*{-0.1cm}

The movements of robots are guided by the routine \textsc{MoveToDestination()}. This routine provides collision-free movements for the robots during the {\it multiplicity creation phase}. This is required to create a unique stable multiplicity point. As discussed above, robots have different types of movements depending on the different phases of the robots. For a robot $r_i\in\mathcal R$, let $H_i(t)$ denotes the set of lines joining two robot positions in $\mathcal R(t)\backslash \{r_i(t)\}$. Algorithm~\ref{alg:move-to-destination} summarizes the routine \textsc{MoveToDestination}(), which invokes the appropriate movement
procedure according to the movement type assigned to robot $r_i$. We describe each of these movements in detail as follows:

 % %\vspace*{-0.5cm}
\textbf{(A)} {\bf Step-in movement:} This movement places a robot, lying on $S(t)$, inside the circle. Suppose robot  $r_i$ wants a {\it step-in movement} w.r.t. the point $p$. Note that $p= \mathcal O(t)$ if $\mathcal R(t)$ is not quasi-regular, otherwise  $p=c_q$. Let us first consider the case when the line segment $\overline{r_i(t)p}$ intersects at least one line from $H_i(t)$ (we exclude lines which coincide with $\overline{r_i(t)p}$ from $H_i(t)$ while computing the intersection points). Let $x_i$ be the nearest of such intersection points to $r_i(t)$ (see Figure~\ref{S-in}(a)). The destination point $d_i$ of $r_i$ is the middle point of the segment $\overline{r_i(t)x_i}$. If none of the lines in $H_i(t)$ intersects $\overline{r_i(t)p}$, then the $d_i$ is the middle point $\overline{r_i(t)p}$ (in this case also, we exclude lines which coincide with $\overline{r_i(t)p}$ from $H_i(t)$ while computing the intersection points)(see Figure~\ref{S-in}(b)). Robot $r_i$ moves towards $d_i$ along the line segment $\overline{r_i(t)d_i}$. Algorithm~\ref{alg:step-in} summarizes the \textsc{StepIn} procedure.

\begin{figure}[htbp]
    \centering

    \begin{subfigure}[t]{0.40\textwidth}
        \centering
        \includegraphics[width=\textwidth]{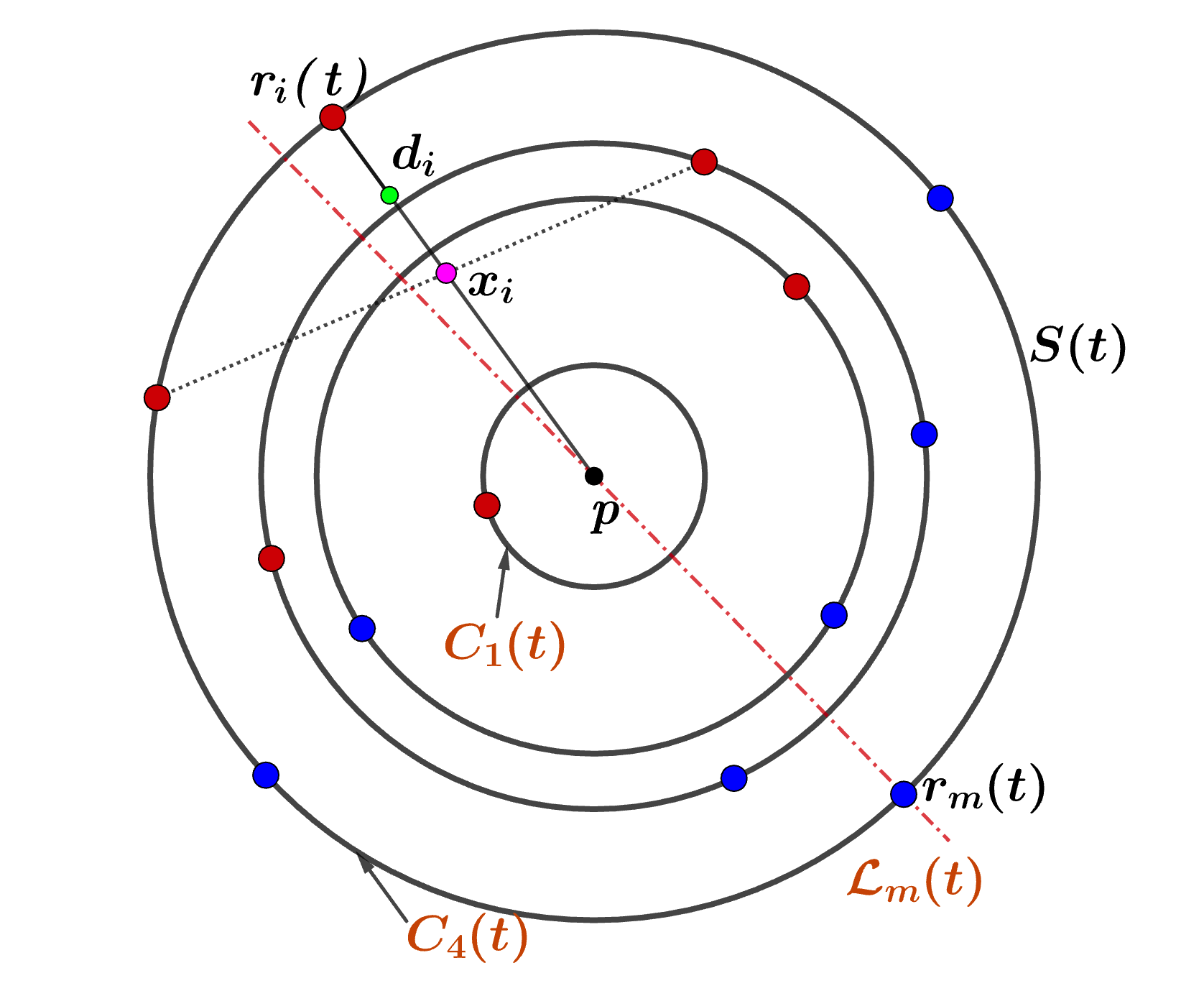}
        \caption{}
        
    \end{subfigure}
   \hspace{0.04\textwidth}%
    \begin{subfigure}[t]{0.40\textwidth}
        \centering
        \includegraphics[width=1.1\textwidth]{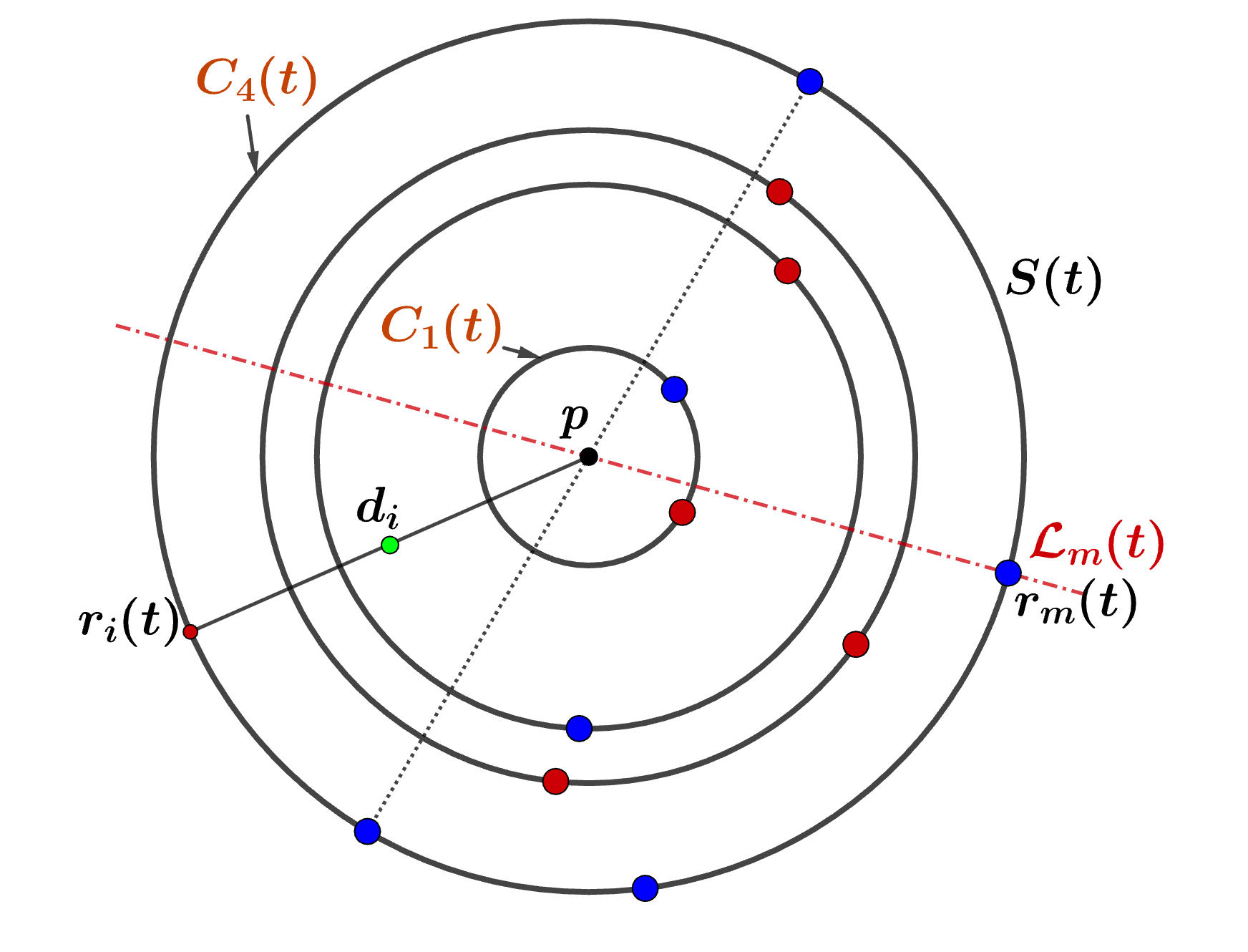}
        \caption{}
        
    \end{subfigure}

\caption{Step-in movement for computing $d_i$: (a) The line segment $\overline{r_i(t)p}$ intersects the dotted line in $H_i(t)$, (b) None of the line in $H_i(t)$ intersect $\overline{r_i(t)p}$.}
\label{S-in}
\end{figure}

\begin{algorithm}[t]
\caption{: \textsc{StepIn}$(r_i,p)$}
\label{alg:step-in}
\begin{algorithmic}[1]
\Require Robot $r_i$ and reference point
$p\in\{\mathcal O(t),c_q\}$

\State Compute $H_i(t)$, the set of lines determined by pairs of
positions in $\mathcal R(t)\setminus\{r_i(t)\}$

\State Remove from $H_i(t)$ every line coincident with
$\overline{r_i(t)p}$

\State $X_i\gets
\left\{
x\in\overline{r_i(t)p}:
x\in h
\text{ for some }h\in H_i(t)
\right\}$

\If{$X_i\neq\emptyset$}
    \State Let $x_i\in X_i$ be nearest to $r_i(t)$
    \State $d_i\gets$ midpoint of $\overline{r_i(t)x_i}$
\Else
    \State $d_i\gets$ midpoint of $\overline{r_i(t)p}$
\EndIf

\State \textbf{move toward} $d_i$ along
$\overline{r_i(t)d_i}$
\end{algorithmic}
\end{algorithm}

\textbf{(B)} {\bf Step-aside movement:} We describe the {\it step-aside movement} of a robot $r_i$ w.r.t. a point $d$. Note that robot $r_i$ have a {\it step-aside movement} in the following two cases: (i) the open line segment $(r_i(t), d_i)$ contains at least one robot position, where $d_i$ is the destination point of $r_i$ and (ii) when $\mathcal R(t)$ is not free-path quasi-regular and $c_q = \mathcal O(t)$. Thus, we have $d\in\{p_m,\mathcal O(t)\}$. Since robots in $\mathcal R$ can not detect the overlapping of {\it the gathering phase} and {\it the formation phase}, we have to be careful about designing {\it step-aside movements} of the robots. Until the {\it formation phase} starts, if two robots from two groups have {\it step-aside movements} in the same round, then these movements are w.r.t. to the same point (we assure this by properly designing the movements of the robots during the {\it multiplicity creation phase} and {\it formation phase}). However, when there is an overlap of {\it the gathering phase} and {\it the formation phase}, then two robots from two groups have two types of movements:  {\it step-aside movements} and {\it step-out movements} w.r.t. different points. Since robots can not distinguish this overlap, we carefully design these movements.

\begin{figure}[htbp]
    \centering

    \begin{subfigure}[t]{0.40\textwidth}
        \centering
        \includegraphics[width=\textwidth]{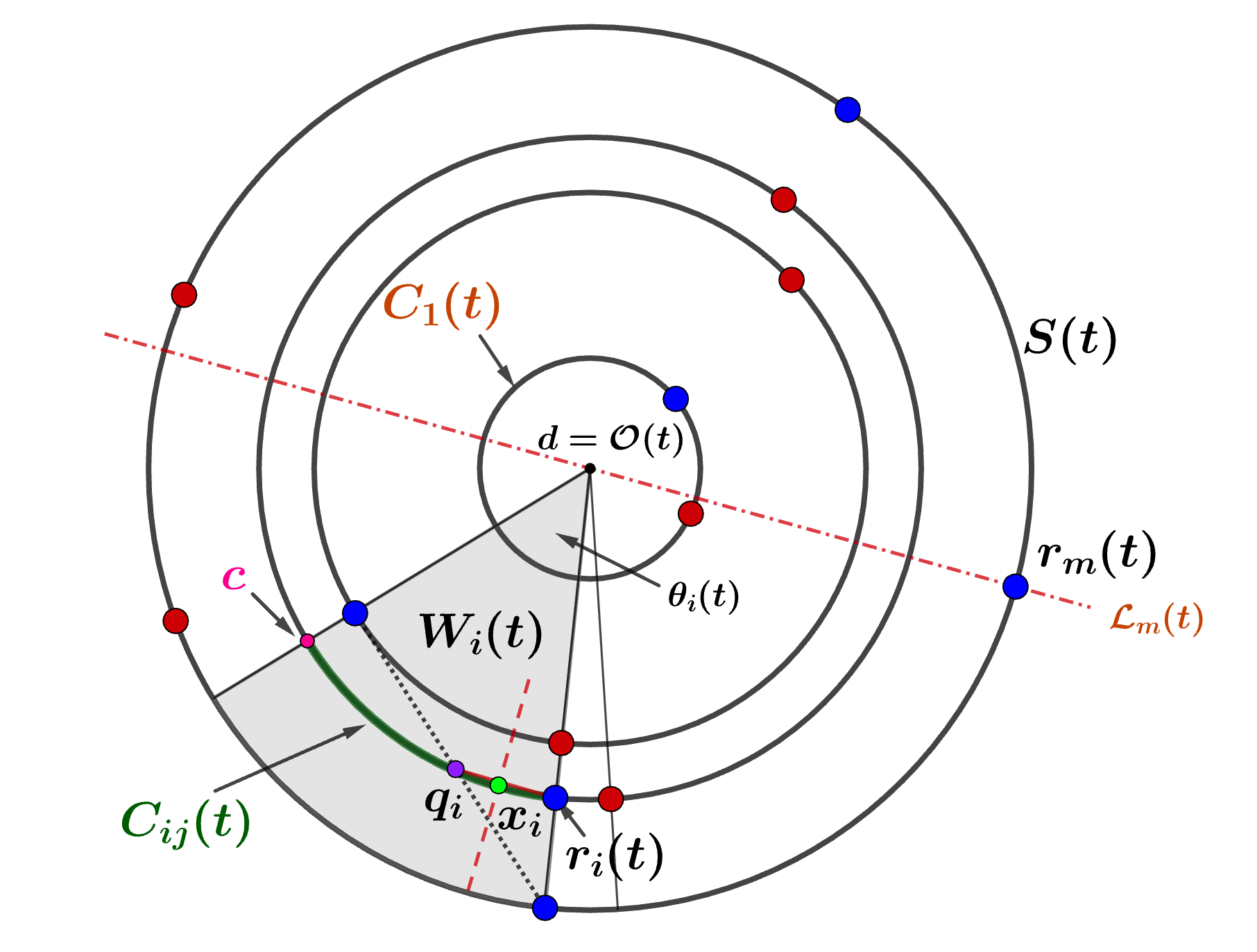}
        \caption{}
        
    \end{subfigure}
   \hspace{0.04\textwidth}%
    \begin{subfigure}[t]{0.40\textwidth}
        \centering
        \includegraphics[width=\textwidth]{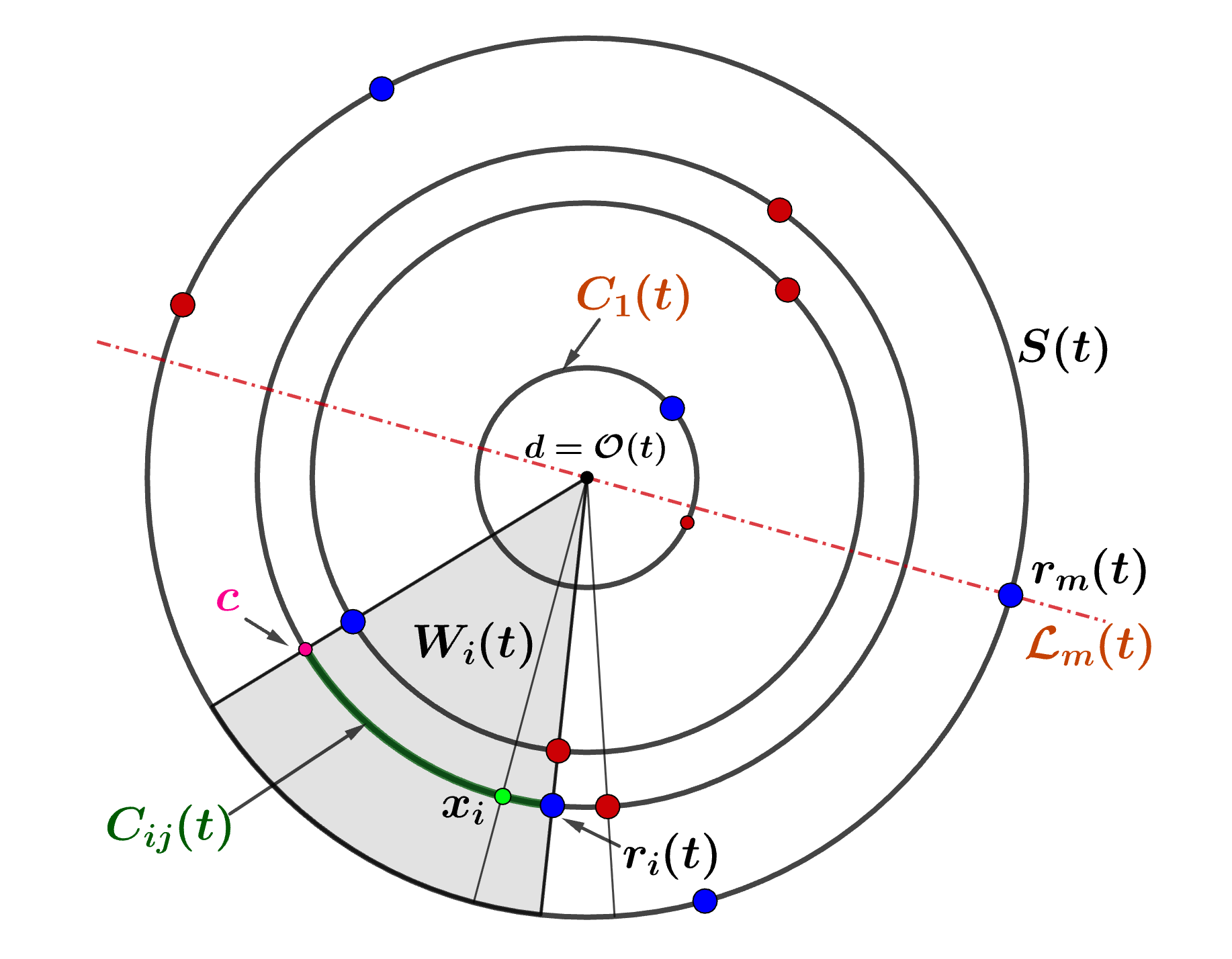}
        \caption{}
        
    \end{subfigure}

   \caption{Step-aside movement for computing $x_i$ when $p_m=\mathcal{O}(t)$: (a) Dotted lines in $H_i(t)$ intersect the arc $C_{ij}(t)$ of circle $C_{l}(t)$ (with $l=3$), (b) None of the lines in $H_i(t)$ intersect $C_{ij}(t)$.}
   
\label{S-aside-1}
\end{figure}

 First, consider the case when $p_m=\mathcal O(t)$. This implies that $d=\mathcal O(t)$ for all robots in $\mathcal R$. The robots in $\mathcal R_g$ need to check $p_m=\mathcal O(t)$. According to our algorithm, if a robot in $\mathcal R_f$ has a {\it step-aside movement}, then its movement is w.r.t. $\mathcal O(t)$. Thus, the robots in $\mathcal R_f$ have {\it step-aside movements} as described in this case (the robots in $\mathcal R_g$ have different movement strategies when $p_m\neq\mathcal O(t)$ to guarantee finite time reachability of the robots in $\mathcal R_g$ to $p_m$). Let robot $r_i$ lie on the circle $C_l(t)$. The destination point $x_i$ of $r_i$ lies on $C_{l}$ and it is computed as follows (see Figure~\ref{S-aside-1}): let $B_i(t)$ be the set of robot positions in $\mathcal R(t)$ which do not lie on $\overline{r_i(t)\mathcal O(t) }$. Let $\overline{r_j(t)\mathcal O(t)}$ and $\overline{r_k(t)\mathcal O(t)}$ be the clockwise and counterclockwise neighbors of $\overline{r_i(t)\mathcal O(t)}$. Consider the angle $\theta_i(t)= max \{\angle{r_i(t)dr_j(t)},\angle{r_i(t)dr_k(t)}\}$ (tie, if any, broken arbitrarily). Without loss of generality suppose, $\theta_i(t)=\angle{r_i(t)dr_j(t)}$.  Let $c$ be  the intersection point between $C_l(t)$ and $rad_j(t)$. Let $W_i(t)$ be the wedge defined by the angle $\theta_i(t)$ and $C_{ij}(t)$ be the arc of $C_l(t)$ which lies in the wedge $W_i(t)$. If at least one line in $H_i(t)$ intersects $C_l(t)$, then let $q_i$ be the nearest to $r_i(t)$ among all such intersection points. In this case,  $x_i$ is the intersection point between $C_{ij}(t)$ and the bisector of the segment $\overline{r_i(t)q_i}$ (see Figure~\ref{S-aside-1}(a)). If none of the lines in $H_i(t)$ intersects $C_{ij}(t)$ (see Figure~\ref{S-aside-1}(b)), then we define $x_i$ to be a point lying on $C_{ij}(t)$ such that $\angle{r_i(t)dx_i}=\frac{1}{3n_i}\theta_i(t)$, where $n_i$ is the number of distinct robot positions on the line segment $\overline{r_i(t)d}$.

\begin{figure}[htbp]
    \centering

    \begin{subfigure}[t]{0.40\textwidth}
        \centering
        \includegraphics[width=\textwidth]{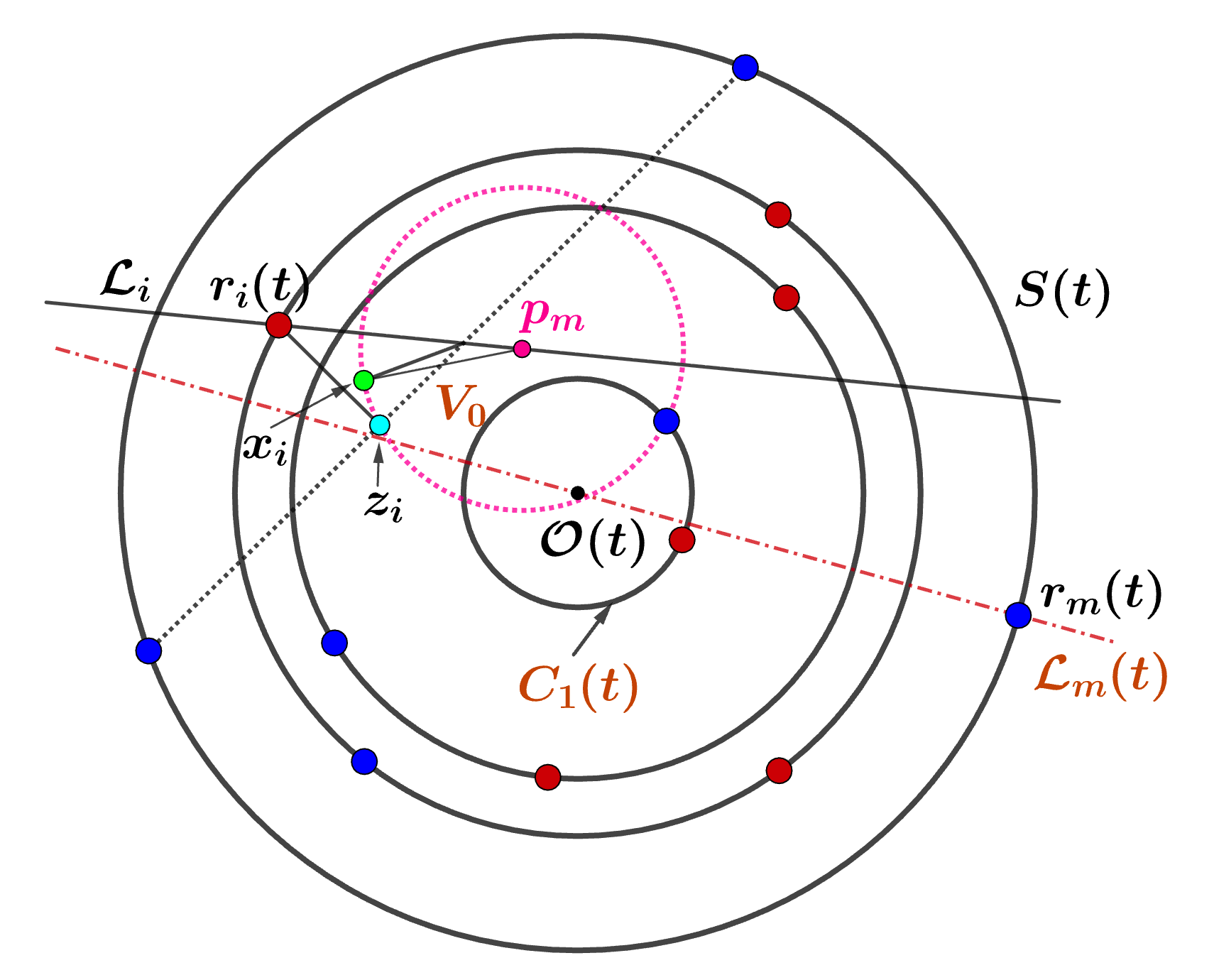}
        \caption{}
        
    \end{subfigure}
   \hspace{0.04\textwidth}%
    \begin{subfigure}[t]{0.40\textwidth}
        \centering
        \includegraphics[width=1.05\textwidth]{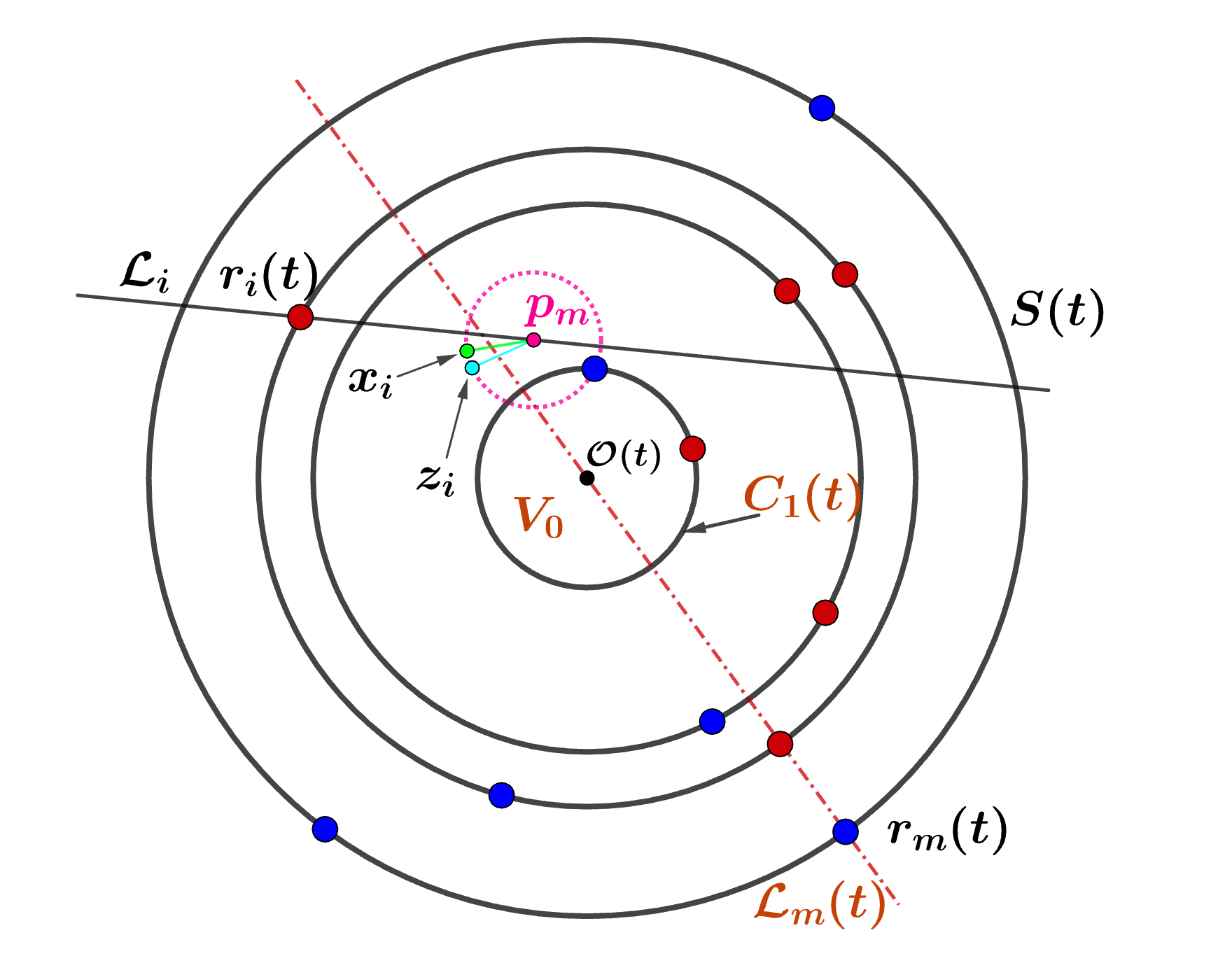}
        \caption{}
        
    \end{subfigure}

   \caption{Step-aside movement for computing $x_i$ when $p_m\neq\mathcal{O}(t)$: (a) Dotted lines in $H_i(t)$ intersect the circle $C^{*}_{k}$, (b) None of the lines in $H_i(t)$ intersect $C^{*}_{k}$.}

\label{S-aside-2}
\end{figure}

% %\vspace*{-0.5cm}

 Note that since we are computing the {\it step-aside movements} w.r.t. $\mathcal O(t)$, we have $\theta_i(t)>0$. Now consider the case when $p_m\neq\mathcal O(t)$ (see Figure~\ref{S-aside-2}). According to the strategy described here, only the robots in $\mathcal R_g$ have {\it step-aside movements}. Let $\mathcal L_i$ be the line passing through $r_i(t)$ and $p_m$. Line $\mathcal L_i$ divides the whole plane into two open halves. We define a half plane $V_o$ delimited by $\mathcal L_i$ as follows: $V_o$ be the half plane delimited by $\mathcal L_i$ which contains $\mathcal O(t)$ if $\mathcal L_i$ does not pass through $\mathcal O(t)$, otherwise $V_o$ is chosen as any one of the two half planes delimited by $\mathcal L_i$.
 Let $D^*(t)=
\left\{
\|p-p_m\| :
p\in\mathcal R(t),\;
p\neq p_m
\right\},$ and let $0<\rho^*_1(t)<\rho^*_2(t)<\cdots<\rho^*_q(t)$ be the distinct values in $D^*(t)$. For each
$k\in\{1,\ldots,q\}$, let $C_k^*(t)$ denote the circle centered at
$p_m$ with radius $\rho_k^*(t)$.
The destination point $x_i$ of $r_i$ lies on $C^*_{k}$ and it is computed as follows: we define a point $z_i$ on $C^*_{k}$ as follows: if at least one line in $H_i(t)$ (excluding the lines coincide with $\mathcal L_i$) intersects $C^*_{k}$. This intersection point lies in $V_o$, then $z_i$ is the nearest of such intersection points to $r_i(t)$ (see Figure~\ref{S-aside-2}(a)). Otherwise $z_i$ is a point on $C^*_{k}\cap V_o$ such that $\angle{z_ip_mr_i(t)}=15^\circ$ (other suitable point may also work)(see Figure~\ref{S-aside-2}(b)). Then $x_i$ is the intersection point between the $C^*_k$ and the bisector of the angle $\angle{r_i(t)p_mz_i}$. Algorithm~\ref{alg:step-aside} summarizes the \textsc{StepAside} procedure.

\begin{figure}[htbp]
    \centering

    \begin{subfigure}[b]{0.30\textwidth}
        \centering
        \makebox[\linewidth][c]{%
            \includegraphics[width=1.2\textwidth]{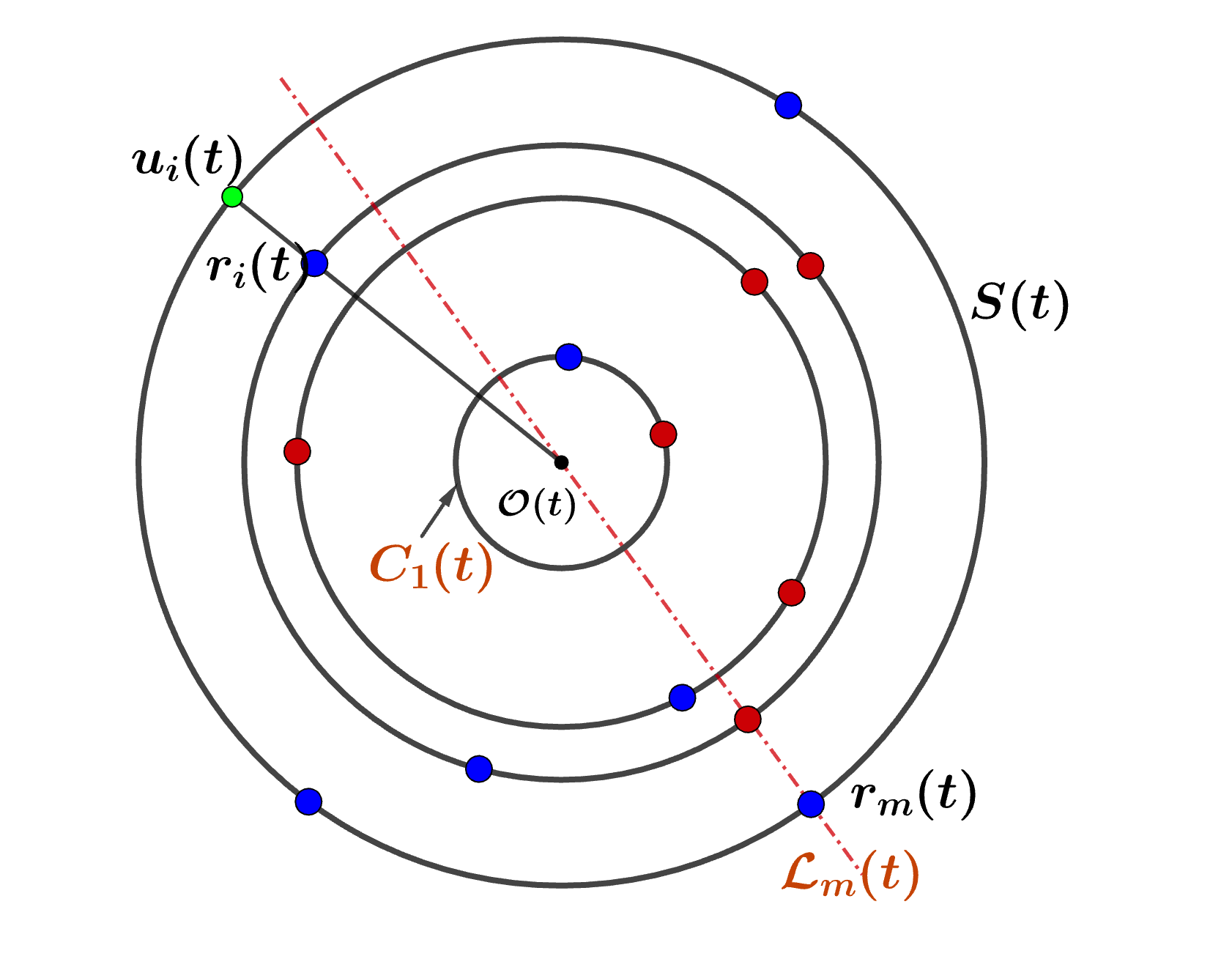}
        }
        \caption{}
    \end{subfigure}
    \hfill
    \begin{subfigure}[b]{0.30\textwidth}
        \centering
        \makebox[\linewidth][c]{%
            \includegraphics[width=1.1\textwidth]{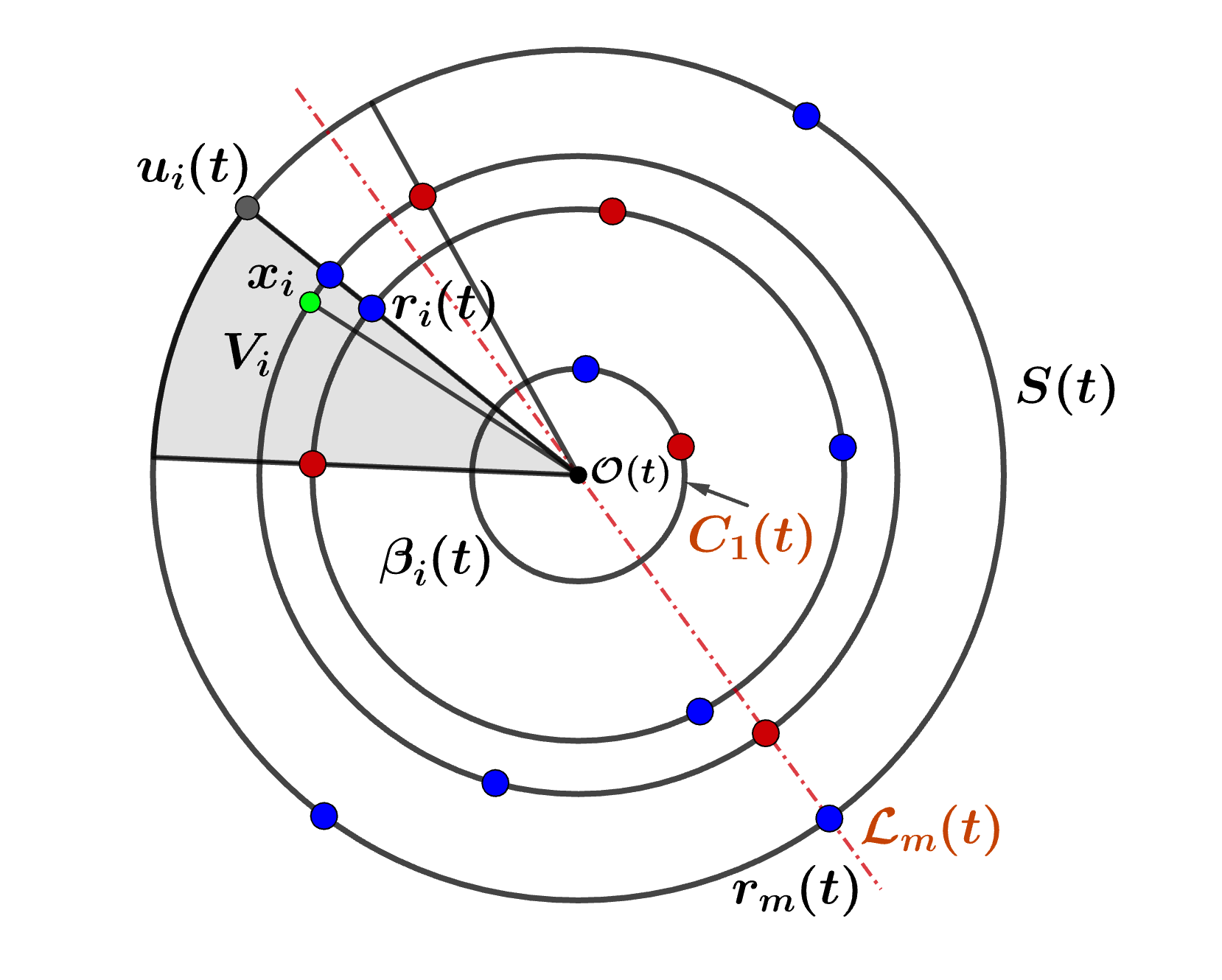}
        }
        \caption{}
    \end{subfigure}
    \hfill
    \begin{subfigure}[b]{0.30\textwidth}
        \centering
        \makebox[\linewidth][c]{%
            \includegraphics[width=\textwidth]{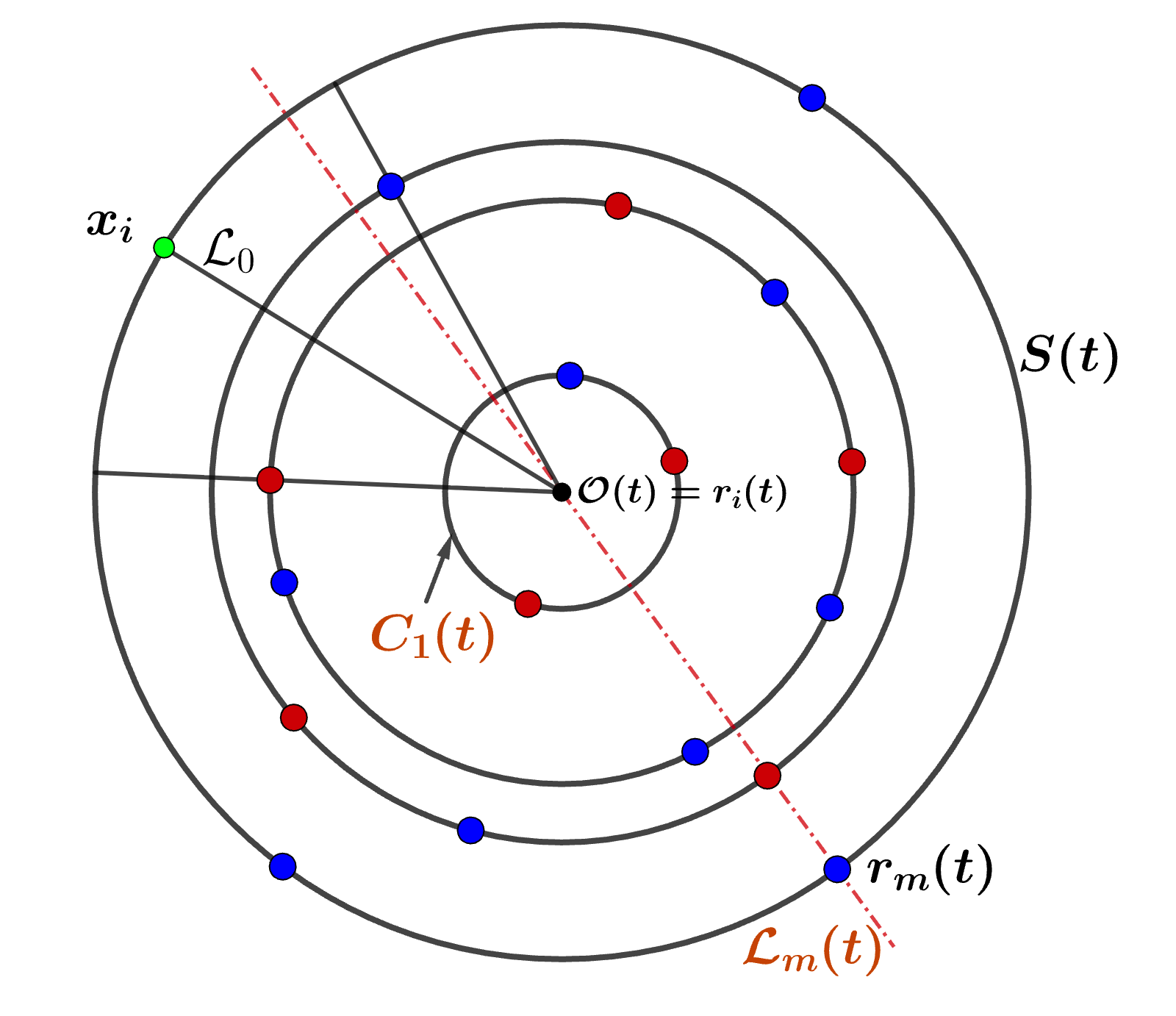}
        }
        \caption{}
    \end{subfigure}

    \caption{Step-out movement for computing $u_i(t)$ and $x_i$:
    (a) The line segment $\overline{r_i(t)u_i(t)}$ contains no other robot
    when $r_i(t)\neq\mathcal O(t)$,
    (b) the segment $\overline{r_i(t)u_i(t)}$ contains multiple robots
    with $r_i(t)\neq\mathcal O(t)$, and
    (c) the case when $r_i(t)=\mathcal O(t)$.}
    \label{S-out}
\end{figure}

\begin{algorithm}[!htbp]
\caption{: \textsc{StepAside}$(r_i,d)$}
\label{alg:step-aside}
\begin{algorithmic}[1]
\Require Robot $r_i$ and reference point
$d\in\{\mathcal O(t),p_m\}$

\If{$p_m=\mathcal O(t)$}
    \State $d\gets\mathcal O(t)$
    \State Let $C_l(t)$ be the radial circle containing $r_i(t)$

    \State Determine the clockwise and counterclockwise neighboring
    radial lines of $rad_i(t)$, containing positions
    $r_j(t)$ and $r_k(t)$, respectively

    \State $    \theta_i(t)\gets
    \max\{
    \angle r_i(t)d r_j(t),
    \angle r_i(t)d r_k(t)
    \}$

    \State Let $r_j(t)$ denote a neighbor attaining the maximum
    \State Let $W_i(t)$ be the corresponding wedge
    \State Let $C_{ij}(t)$ be the arc of $C_l(t)$ contained in $W_i(t)$

    \State Compute all intersections of the lines in $H_i(t)$
    with $C_{ij}(t)$

    \If{at least one such intersection exists}
        \State Let $q_i$ be the intersection nearest to $r_i(t)$
        \State Let $x_i$ be the intersection of $C_{ij}(t)$
        with the perpendicular bisector of
        $\overline{r_i(t)q_i}$
    \Else
        \State Let $n_i$ be the number of distinct occupied positions
        on $\overline{r_i(t)d}$
        \State Choose $x_i\in C_{ij}(t)$ such that $\angle r_i(t)d x_i
        =
        \frac{\theta_i(t)}{3n_i}$
    \EndIf

\Else
    \Comment{$p_m\neq\mathcal O(t)$}

    \State $d\gets p_m$
    \State Let $\mathcal L_i$ be the line through $r_i(t)$ and $p_m$

    \If{$\mathcal O(t)\notin\mathcal L_i$}
        \State Let $V_o$ be the open half-plane bounded by
        $\mathcal L_i$ that contains $\mathcal O(t)$
    \Else
        \State Choose one admissible open half-plane $V_o$
        bounded by $\mathcal L_i$
    \EndIf

    \State Let $C_k^*(t)$ be the radial circle centered at $p_m$
    containing $r_i(t)$

    \State $  Z_i\gets
    \left\{
    z\in C_k^*(t)\cap V_o :
    z\in h,\;
    h\in H_i(t),\;
    h\neq\mathcal L_i
    \right\}$

    \If{$Z_i\neq\emptyset$}
        \State Let $z_i\in Z_i$ be nearest to $r_i(t)$
    \Else
        \State Choose $z_i\in C_k^*(t)\cap V_o$ such that $\angle z_i p_m r_i(t)=15^\circ$
    \EndIf

    \State Let $x_i$ be the intersection of $C_k^*(t)$ with
    the internal angle bisector of
    $\angle r_i(t)p_mz_i$
\EndIf

\State \textbf{move toward} $x_i$
\end{algorithmic}
\end{algorithm}

\textbf{(C)} {\bf Step-out movement:} This movement is taken by robots in $\mathcal R_f$ during the {\it formation phase}. Consider a robot $r_i\in\mathcal R_f$ having position in $S_{in}(t)$. {\it Step-out movements} place robot $r_i$ at a point on the circle $\mathcal S(t)$. First consider the case when $r_i(t)\neq\mathcal O(t)$ (see Figure~\ref{S-out}(A-B)). Let $rad_i(t)$ intersect $S(t)$ at $u_i(t)$ and $r_i(t)$ lies on $C_k(t)$ for some $k$. If the line segment $\overline{r_i(t)u_i(t)}$ does not contain any other robot position, then robot $r_i$ moves towards $u_i(t)$ along the line segment $\overline{r_i(t)u_i(t)}$ (see Figure~\ref{S-out}(a)). Otherwise, robot $r_i$ computes its destination point $x_i(t)$ on $C_{k+1}(t)$ as follows (see Figure~\ref{S-out}(b)). Let $A_i(t)$ be the set of robot positions not lying on $rad_i(t)$. Let $r_j(t), r_k(t)\in A_i(t)$ be such that $rad_j(t)$ and
$rad_k(t)$ are the two clockwise and counterclockwise neighbors of $rad_i(t)$. Let $\beta_i(t) = max\{\angle{r_i(t)\mathcal O(t)r_j(t)},\angle{r_i(t)\mathcal O(t)r_k(t)}\}$ (tie, if any, broken arbitrarily). Without loss of generality, suppose $\beta_i(t)=\angle{r_i(t)\mathcal O(t)r_j(t)}$. Since $S(t)$ contains at least two robot positions on its boundary, $\beta_i(t)\neq 0$. Let $V_i$ be the wedge defined by the angle $\beta_i(t)$. We define $x_i$ to be a point lying on $V_i(t)\cap C_{k+1}(t)$ such that $\angle{r_i(t)\mathcal O(t)x_i(t)}= \frac{1}{3m_i}\beta_i(t)$, $m_i$ is the number of distinct robot positions on the line segment $r_i(t)u_i(t)$.

 Now, consider the case when $r_i(t) = \mathcal O(t)$ (see Figure~\ref{S-out}(c)). Let $r_j(t)$ and $r_l(t)$ be two robot positions not lying at $\mathcal O(t)$ such that $rad_j(t)$ and $rad_l(t)$ are neighbors and  $\angle{r_j(t)\mathcal O(t)r_l(t)}$ is maximum for all such consecutive radial lines (tie, if any, broken arbitrarily). Let $\mathcal L_0$ be the bisector of $\angle{r_j(t)\mathcal O(t)r_l(t)}$. Then $x_i(t)$ is the intersection point between $S(t)$ at $\mathcal L_0$. Algorithm~\ref{alg:step-out} summarizes the \textsc{StepOut} procedure, which determines an outward destination for a robot in
$\mathcal R_f$ during the formation phase. Table~\ref{tab:movement-summary} summarizes the movement rules followed by the robots in $\mathcal R_g$ and $\mathcal R_f$ for the different configuration classes during the multiplicity-creation phase.

 \begin{algorithm}[!htbp]
\caption{: \textsc{StepOut}$(r_i)$}
\label{alg:step-out}
\begin{algorithmic}[1]
\Require Robot $r_i\in\mathcal R_f$

\If{$r_i(t)\neq\mathcal O(t)$}

    \State Let $u_i(t)=rad_i(t)\cap S(t)$

    \If{$(r_i(t),u_i(t))$ contains no occupied position}
        \State \textbf{move toward} $u_i(t)$ along
        $\overline{r_i(t)u_i(t)}$
        \State \Return
    \EndIf

    \State Let $C_k(t)$ be the radial circle containing $r_i(t)$

    \State Determine the clockwise and counterclockwise neighboring
    radial lines of $rad_i(t)$, containing positions
    $r_j(t)$ and $r_\ell(t)$, respectively

    \State $\beta_i(t)\gets
    \max\{
    \angle r_i(t)\mathcal O(t)r_j(t),
    \angle r_i(t)\mathcal O(t)r_\ell(t)
    \}$
    \State Let $V_i(t)$ be the wedge corresponding to a neighbor
    attaining $\beta_i(t)$

    \State Let $m_i$ be the number of distinct occupied positions on
    $\overline{r_i(t)u_i(t)}$

    \State Choose $x_i(t)\in V_i(t)\cap C_{k+1}(t)$ such that $\angle r_i(t)\mathcal O(t)x_i(t)
    =
    \frac{\beta_i(t)}{3m_i}$
    \State \textbf{move toward} $x_i(t)$

\Else
    \Comment{$r_i(t)=\mathcal O(t)$}

    \State Find two consecutive radial lines
    $rad_j(t)$ and $rad_\ell(t)$ whose enclosed angle is maximum

    \State Let $\mathcal L_0^{+}$ be the ray from $\mathcal O(t)$
    bisecting this maximum angular sector

    \State $x_i(t)\gets
    \mathcal L_0^{+}\cap S(t)$

    \State \textbf{move toward} $x_i(t)$
\EndIf
\end{algorithmic}
\end{algorithm}

\begin{table}[htbp]
\centering
\small
\begin{threeparttable}

\caption{Summary of movement rules during the multiplicity-creation phase.}
\label{tab:movement-summary}

\renewcommand{\arraystretch}{1.15}

\begin{tabularx}{\textwidth}{
@{}
>{\raggedright\arraybackslash}p{2.9cm}
>{\raggedright\arraybackslash}p{2.9cm}
>{\raggedright\arraybackslash}X
>{\raggedright\arraybackslash}X
@{}
}
\toprule
Configuration class
& Sub-case
& $\mathcal R_g$ robots
& $\mathcal R_f$ robots \\
\midrule

\2
& free-path Q-regular
& direct $\to c_q$
& stationary \\

\2
& Q-regular, not free-path
& direct/step-aside $\to \mathcal O(t)$
& stationary \\

\2
& not quasi-regular\tnote{a}
& direct/step-aside $\to \mathcal O(t)$
& stationary \\

\3
& free-path Q-regular
& direct $\to c_q$
& stationary \\

\3
& Q-regular, $c_q\neq\mathcal O(t)$
& stationary (non-pivotal step-in)
& stationary (non-pivotal step-in) \\

\3
& Q-regular, $c_q=\mathcal O(t)$
& step-aside w.r.t.\ $\mathcal O(t)$ if blocked
& step-aside w.r.t.\ $\mathcal O(t)$ if blocked \\

\3
& not quasi-regular
& step-in w.r.t.\ $\mathcal O(t)$ (non-pivotal)
& step-in w.r.t.\ $\mathcal O(t)$ (non-pivotal) \\

\bottomrule
\end{tabularx}

\begin{tablenotes}
\footnotesize
\item[a]
Does not occur for linear configurations when $n\geq 8$;
see Observation~\ref{obs:linear}.
\end{tablenotes}

\end{threeparttable}
\end{table}

\subsection {Description of Algorithm \textsc{PatternFormation()}}
 In this section, we describe our proposed algorithm. The execution steps for the robots are described separately for each of the three phases. Algorithm~\ref{alg:pattern-formation} summarizes the overall control flow of \textsc{PatternFormation}(), directing each robot to the appropriate phase according to the current configuration.
 \begin{enumerate}[(A)]
 
 \item {\bf The multiplicity creation phase:} During this phase, our algorithm exploits the number of robot positions in $S_{in}(t)$ and the symmetry of $\mathcal R(t)$. Algorithm~\ref{alg:multiplicity-creation} summarizes the \textsc{MultiplicityCreation} procedure, which determines the movement of a robot according to the configuration class and its quasi-regularity properties. We have the following observation.

\begin{observation}[Linear configurations are quasi-regular but not free-path quasi-regular]
\label{obs:linear}
Let $\mathcal R(t)\in\mathsf{Dense}$ be a linear configuration, i.e., all robot
positions lie on a common line $\mathcal L(t)$, and let
$n=|\mathcal R(t)|\ge 8$. Then $\mathcal R(t)$ is quasi-regular with a unique
centre of quasi-regularity $c_q$, and $\mathcal R(t)$ is \emph{not}
free-path quasi-regular.
\end{observation}

\begin{algorithm}[!htbp]
\caption{: \textsc{PatternFormation}$(r_i)$}
\label{alg:pattern-formation}
\begin{algorithmic}[1]

\If{$r_i\in\mathcal R_g$ and a multiplicity point $p_m$ exists}
    \State \Call{Gathering}{$r_i,p_m$}
    \State \Return
\EndIf

\If{$r_i\in\mathcal R_f$ and
$|\mathcal R(t)|\leq N_f+1$}
    \State \Call{Formation}{$r_i$}
    \State \Return
\EndIf

\If{$r_i\in\mathcal R_f$ and $r_i(t)$ is a multiplicity point}
    \State \textbf{remain stationary}
    \State \Return
\EndIf

\State \Call{MultiplicityCreation}{$r_i$}

\end{algorithmic}
\end{algorithm}

\begin{algorithm}[!htbp]
\caption{: \textsc{MultiplicityCreation}$(r_i)$}
\label{alg:multiplicity-creation}
\begin{algorithmic}[1]

\State By default, $r_i$ remains stationary.

\If{$\mathcal R(t)\in$\2}

    \If{$\mathcal R(t)$ is free-path quasi-regular}
        \If{$r_i\in\mathcal R_g$ and $r_i(t)\in S_{in}(t)$}
            \State \textbf{move directly toward} $c_q$
        \EndIf

    \ElsIf{$r_i\in\mathcal R_g$ and $r_i(t)\in S_{in}(t)$}
        \If{$(r_i(t),\mathcal O(t))$ contains no robot position}
            \State \textbf{move directly toward} $\mathcal O(t)$
        \Else
            \State \Call{StepAside}{$r_i,\mathcal O(t)$}
        \EndIf
    \EndIf

\ElsIf{$\mathcal R(t)\in$\3}

    \If{$\mathcal R(t)$ is free-path quasi-regular}
        \If{$r_i\in\mathcal R_g$}
            \State \textbf{move directly toward} $c_q$
        \EndIf

    \ElsIf{$\mathcal R(t)$ is quasi-regular}

        \If{$c_q\neq\mathcal O(t)$}
            \State Compute the pivotal set $\mathcal P'(t)$
            \If{$r_i(t)\in S_{out}(t)$ and
            $r_i(t)\notin\mathcal P'(t)$}
                \State \Call{StepIn}{$r_i,c_q$}
            \EndIf

        \Else
            \State Let $x_i=rad_i(t)\cap S(t)$
            \If{$r_i(t)\in S_{in}(t)$ and
            ($x_i$ is occupied \textbf{or}
            $(r_i(t),\mathcal O(t))$ contains a robot position)}
                \State \Call{StepAside}{$r_i,\mathcal O(t)$}
            \EndIf
        \EndIf

    \Else
        \State Compute the pivotal set $\mathcal P'(t)$
        \If{$r_i(t)\in S_{out}(t)$ and
        $r_i(t)\notin\mathcal P'(t)$}
            \State \Call{StepIn}{$r_i,\mathcal O(t)$}
        \EndIf
    \EndIf

\EndIf

\end{algorithmic}
\end{algorithm}
 
\begin{proof}
Since the robot positions in $\mathcal R(t)$ are collinear, they admit a
total order along $\mathcal L(t)$, computable by every robot without any
axis agreement (it only requires comparing distances along the single
line on which all positions lie).
 
\smallskip
\noindent\textbf{Existence.}
Let $c\in\mathcal L(t)$ be the unique balanced split point of
$\mathcal R(t)$ along $\mathcal L(t)$: if $n$ is even, $c$ is the midpoint
of the two median consecutive positions $r_{(n/2)}, r_{(n/2+1)}$; if $n$
is odd, $c$ is the median robot position $r_{(\lceil n/2\rceil)}$ itself.
In either case, taking $\mathcal B(t)=\mathcal R(t)$ (even case) or
$\mathcal B(t)=\mathcal R(t)\setminus\{c\}$ (odd case), exactly
$\lfloor n/2\rfloor$ positions of $\mathcal B(t)$ lie on each side of $c$
along $\mathcal L(t)$, all robots on a given side sharing one ray from
$c$. Hence
\[
SA(\mathcal B(t),c)=X^2,\qquad X=\big(\underbrace{0,\dots,0}_{\lfloor n/2\rfloor-1},\pi\big),
\]
so $\mathcal R(t)$ is quasi-regular with centre $c_q=c$.
 
\smallskip
\noindent\textbf{Uniqueness of $c_q$.}
Any other point $c'\in\mathcal L(t)$, $c'\ne c_q$, changes the count of
robot positions on at least one side of the split (since all positions
are distinct and $c_q$ is the unique balanced split point), so $c'$
cannot yield $SA(\cdot,c')=X^2$. Hence $c_q$ is unique, and every robot
can identify it consistently despite full disorientation.
 
\smallskip
\noindent\textbf{Failure of the free-path property.}
Since $c_q\in\mathcal L(t)$, every $qrad_i(t)$ is a sub-ray of
$\mathcal L(t)$ itself, so all positions in $\mathcal R(t)\setminus\{c_q\}$
lie on exactly two rays from $c_q$, each containing
$\lceil (n-1)/2\rceil \ge \lceil 7/2\rceil = 4$ robot positions (using
$n\ge 8$). Hence for the farthest position $r_i(t)$ on a ray, every other
position on that ray lies on $qrad_i(t)$ strictly between $c_q$ and
$r_i(t)$. By Definition~4, $\mathcal R(t)$ is not free-path quasi-regular.
\end{proof}
 
\begin{remark}
Observation~\ref{obs:linear} shows that for \2
configurations with $n\ge 8$, a linear configuration always falls under
sub-case~(1.b) (quasi-regular but not free-path quasi-regular), whose
movement rule is defined with respect to $\mathcal O(t)$ rather than
$c_q$, and hence applies verbatim without modification. Consequently the
qualifier ``non-linear'' in the statements of sub-cases~(1.a) and~(1.b)
in Section~4.5 is unnecessary and is removed (see Part~B, item~B.1).
\end{remark}

{\bf Case 1} $\boldsymbol{\mathcal R(t) \in}$  \2: Since $|\mathcal R_g|\ge6$ and $|S_{out}(t)|\leq 4$, there are at least two robot positions from $\mathcal R_g(t)$ which lies in $S_{in}(t)$. Two robots are sufficient to create a multiplicity point and thus, the robots having positions in $S_{in}(t)$ can create a multiplicity point within a finite time. We need to coordinate their movement so that no more than one multiplicity point is created during this process, and the process creates the multiplicity point in finite time. In order to achieve so, we consider two cases separately as follows:

\textbf{(1.a) $\mathcal R(t)$ is free-path quasi-regular:} Since robots in $\mathcal R_f$ do not have global multiplicity detection capability, they can recognize $\mathcal R(t)$ quasi-regular configurations until at least two robots reach $c_q$, the centre of the quasi-regularity. Once two robots reach $c_q$, a multiplicity point is created at $c_q$. Since robots have free paths towards $c_q$, they reach this point by direct movements, and no other multiplicity point is created during these movements. The robots act in this case as follows: if $r_i(t)\in S_{in}(t)\cap \mathcal R_g(t)$, the centre $c_q$ is the destination point for $r_i$ and robot $r_i$ has a direct movement towards $c_q$. In the rest of the cases, robots do not move. Since robots in $\mathcal R_f$ do not move, if the centre $c_q$ contains a robot from $\mathcal R_f$, then the multiplicity point contains a robot from $\mathcal R_f$. To make this multiplicity point stable, our approach does not move this robot from $\mathcal R_f$ until at least two robots from $\mathcal R_g$ reach $c_q$ (this case is handled during the {\it formation phase} using the local weak multiplicity detection capability of robots in $\mathcal R_f$).

\textbf{(1.b) $\mathcal R(t)$ is quasi-regular but not free-path quasi-regular or $\mathcal R(t)$ is not quasi-regular:} We maintain $S(t)$, until a multiplicity point is created. Our approach does not move the robots lying on $S(t)$. Otherwise, if $r_i(t)\in S_{in}(t)\cap \mathcal R_g(t)$, then $\mathcal O(t)$ is the destination point of $r_i$. In the rest of the cases, robots do not move. A robot has either a {\it direct movement} or {\it step-aside movement} w.r.t. $\mathcal O(t)$ depending on its position.

% %\vspace*{-0.2cm}
{\bf Case 2} $\boldsymbol{\mathcal R(t)\in}$ \3: In this case, $\mathcal R(t)$ can not be linear. The movement of a robot $r_i$ depends on the symmetry of $\mathcal R(t)$ as follows:

\textbf{(2.a) $\boldsymbol{\mathcal R(t)}$ is free-path quasi-regular:} Our approach maintains the quasi-regularity until at least two robots reach $c_q$. A robot $r_i\in \mathcal R_g$ moves towards $c_q$ following a direct movement. If a robot belongs to $\mathcal R_f$, it does not move. In this case, robots have only direct movements with respect to $c_q$.

\textbf{(2.b) $\boldsymbol{\mathcal R(t)}$  is quasi-regular but not free-path quasi-regular:} In this case, robots move either to convert the current configuration to a configuration in \2 or to create free paths to $\mathcal O(t)$ for robots. The approach depends on the positions of $\mathcal O(t)$ and $c_q$. First, consider the case when $c_q\neq \mathcal O(t)$. Robot $r_i$ computes $\mathcal P'(t)$, that is, the set of pivotal robots. If $r_i(t)\in S_{out}(t)$ and $r_i(t)\notin \mathcal P'(t)$, robot $r_i$ computes a  destination point on $\overline{r_i(t)c_q}$ according to the {\it step-in movement} and moves towards this point. Otherwise, it does not move. Now suppose $c_q = \mathcal O(t)$. Consider a robot $r_i$ such that $r_i(t)\in S_{in}(t)$. Let $rad_i(t)$ intersect $S(t)$ at the point $x_i$. If $x_i$ contains a robot position or the open line segment $(r_i(t), \mathcal O(t))$ contains at least one robot position, then robot $r_i$ takes a {\it step-aside movement} w.r.t. $\mathcal O(t)$ (robot $r_i$ moves out of the line segment $\overline{r_i(t)\mathcal O(t)}$).  In the rest of the cases, robot $r_i$ does not move.

\textbf{(2.c) $\boldsymbol{\mathcal R(t)}$  is not quasi-regular:} In this case, we try to convert the current configuration into a configuration in \2 and in order to do so, we maintain $S(t)$. If $r_i(t)\in S_{out}(t)$ and $r_i(t)\notin \mathcal P'(t)$, then robot $r_i$ has a {\it step-in movement} w.r.t. $\mathcal O(t)$ (robot $r_i$ moves to a point on $\overline{r_i(t)\mathcal O(t)}$). Otherwise, it does not move.

\begin{algorithm}[ht]
\caption{: \textsc{Gathering}$(r_i,p_m)$}
\label{alg:gathering}
\begin{algorithmic}[1]
\Require Robot $r_i\in\mathcal R_g$ and stable multiplicity point $p_m$

\If{$r_i(t)=p_m$}
    \State \textbf{remain stationary}
    \State \Return
\EndIf

\If{$(r_i(t),p_m)$ contains no robot position}
    \State \textbf{move directly toward} $p_m$
    \State \Return
\EndIf

\If{$\mathcal R(t)\in$\3 and
$\mathcal R(t)$ is quasi-regular but not free-path quasi-regular and $c_q=\mathcal O(t)$}
    \State \textbf{remain stationary}
\Else
    \State \Call{StepAside}{$r_i,p_m$}
\EndIf

\end{algorithmic}
\end{algorithm}

\item {\bf The gathering phase:} Robots in $\mathcal R_g$ executes this phase when a multiplicity point $p_m$ is created. However, not all robots in $\mathcal R_f$ can detect the execution of this phase. If a robot in $\mathcal R_f$ lies at the multiplicity point $p_m$, it does not move (robots in $\mathcal R_f$ have local weak multiplicity capability). Otherwise, the robots in $\mathcal R_f$ continue acting in the same way as they do during the {\it multiplicity creation phase}. Now, consider a robot $r_i\in\mathcal R_g$. If robot $r_i(t)= p_m$, then robot $r_i$ does not move. Otherwise,  robot $r_i$ moves towards the multiplicity point $p_m$ in the following way:

\textbf{(B.1)} If the line segment $(r_i(t),p_m)$ does not contain any other robot position, then robot $r_i$ has a direct movement towards $p_m$.

\textbf{(B.2)} Otherwise, the line segment $(r_i(t),p_m)$  contains at least one robot position.  If $\mathcal R(t)\in$\3 and  $\mathcal R(t)$  is quasi regular but not free-path quasi regular with $c_q=\mathcal O(t)$, then $r_i$ does not move. Otherwise,  robot $r_i$ has a {\it step-aside movement} w.r.t. $p_m$.

Algorithm~\ref{alg:gathering} summarizes the \textsc{Gathering} procedure, which directs a robot in $\mathcal R_g$ toward the stable multiplicity point $p_m$ using either
a direct or a step-aside movement, depending on the current configuration.

\item {\bf The formation phase:} This phase is executed only by the robots in $\mathcal R_f$ when the total number of distinct robot positions in the system is at most $N_f+1$, i.e., $|\mathcal R(t)|\le N_f+1$. Since the multiplicity point $p_m$ may contain a robot from $\mathcal R_f$, a robot from $\mathcal R_g$ may not lie at $p_m$. Since the robots in $\mathcal R_g$ do not know the size of $\mathcal R_g$, they can not determine the start of this phase. Thus, exactly one robot in $\mathcal R_g$ may execute the {\it gathering phase} while the robots in $\mathcal R_f$ execute the {\it formation phase} simultaneously. However, the robots are unaware of this overlap. Thus, the main challenge here is to avoid collisions between the robots in $\mathcal R_f$. Note that there is a special case in which we can not avoid a collision between a robot from $\mathcal R_g$ and a robot from $\mathcal R_f$. However, this multiplicity point is unstable (once the robot from $\mathcal R_f$ moves, this multiplicity is broken). Robots in $\mathcal R_f$ move in some order. 
\begin{algorithm}[ht]
\caption{: \textsc{Formation}$(r_i)$}
\label{alg:formation}
\begin{algorithmic}[1]
\Require Robot $r_i\in\mathcal R_f$ and
$|\mathcal R(t)|\leq N_f+1$

\If{$r_i(t)$ is a multiplicity point}

    \If{$r_i(t)\notin\partial S(t)$}
        \State \Call{StepOut}{$r_i,S(t)$}
    \Else
        \State Let $y_i$ be the midpoint of
        $\overline{r_i(t)\mathcal O(t)}$
        \State \textbf{move toward} $y_i$
    \EndIf

    \State \Return
\EndIf

\State Let $C_k(t)=S(t)$

\If{$r_i(t)\in S(t)$}
    \State \textbf{remain stationary}

\ElsIf{$r_i(t)\in C_{k-1}(t)$}

    \State Let
    $x_i=rad_i(t)\cap S(t)$

    \If{$x_i$ is not occupied}
        \State \textbf{move directly toward} $x_i$
    \Else
        \State \Call{StepOut}{$r_i,S(t)$}
    \EndIf

\ElsIf{$r_i(t)\in C_{k-2}(t)$ and
$C_{k-1}(t)$ contains exactly one occupied position}

    \State \Call{StepOut}{$r_i,C_{k-1}(t)$}

\Else
    \State \textbf{remain stationary}
\EndIf

\end{algorithmic}
\end{algorithm}

First, consider the case when a robot $r_i\in \mathcal R_f$ does not lie at a multiplicity point. Let $C_k=S(t)$ for some $k$. Robot $r_i$  moves according to one of the following ways:

\textbf{(C.1)} If robot $r_i$ lies on $S(t)$, it does not move. 

\textbf{(C.2)} Suppose $r_i$ lies on $C_{k-1}(t)$. If $x_i$, the intersection point between $S(t)$ and $rad_i(t)$, does not contain any robot position, then $r_i$ has a direct movement towards $x_i$. Otherwise, robot $r_i$ moves according to the  {\it step-out movement} w.r.t. $S(t)$.

\textbf{(C.3)} Suppose robot $r_i$ lies on $C_{k-2}(t)$ and $C_{k-1}(t)$ contains exactly one robot position. In this case, robot $r_i$ moves towards the circle $C_{k-1}(t)$  by a {\it step-out movement} w.r.t. the circle $C_{k-1}(t)$.

\textbf{(C.4)} In the rest of the cases, robot $r_i$ does not move.

Next, consider the case when robot $r_i$ lies at a multiplicity point. In this case, robot $r_i$ moves even if it lies on $S(t)$ (the multiplicity point may lie on $S(t)$). Robot $r_i$ moves according to the {\it step-out movement} w.r.t. $S(t)$, if it does not lie on the boundary of $S(t)$. Otherwise, robot $r_i$ first moves towards the middle point of the line segment, $\overline{r_i(t)\mathcal O_t}$ and then it moves according to the {\it step-out movement} w.r.t. $S(t)$. The formation rules for robots in $\mathcal R_f$ are summarized in Algorithm~\ref{alg:formation}.
%\vspace*{-0.5cm}
 \end{enumerate}

\section{Correctness of Algorithm \textsc{PatternFormation()}}

In this section, we prove the correctness of Algorithm~\textsc{PatternFormation()} by establishing the required properties through a sequence of lemmas and concluding with the main theorem.
\begin{lemma}
\label{lemma-mcp-1}
Routine \textsc{MoveToDestination()} provides collision-free movements for the robots during the multiplicity creation phase.
\end{lemma}

\begin{proof} 
%\vspace*{-0.3cm}
Consider two robots $r_i, r_j\in R$ which are to move in round t. Two robots do not collide if their movement paths do not intersect each other except at the desired destination point. We consider each of the phases separately :

\textbf{{\bf Case 1 $\boldsymbol{\mathcal R(t)\in}$ \2:}} In this case, $\mathcal O(t)$ is the destination point for all the robots in $\mathcal R_g$ and the robots in $\mathcal R_f$ do not move. Robots have two types of movements: the {\it direct movement} and these movements are w.r.t. $\mathcal O(t)$ only. Consider $rad_i(t)$ and $rad_j(t)$. First suppose that $rad_i(t)$ and $rad_j(t)$ coincide. Let $r_i$ and $r_j$ lie on the circles $C_k(t)$ and $C_l(t)$ respectively. Without loss of generality, suppose $l<k$. In this case, at most one of $r_i$ and $r_j$ can have direct movement. Now, in round $t$, consider all possible combinations of movements of robots $r_i$ and $r_j$. In all of them, the paths of movement of these two robots are completely separated by the circle $C_l(t)$. Thus,  robots $r_i$ and $r_j$ do not collide.

 Next suppose that $rad_i(t)$ and $rad_j(t)$ are distinct. Let $\mathcal L^*$ be the bisector of the angle $\angle{r_i(t)\mathcal O(t)r_j(t)}$. In this case, all possible paths of the two robots are separated by the bisector $\mathcal L^*$. 

 \textbf{{\bf Case 2 $\boldsymbol{\mathcal R(t)\in}$ \3:}} The movements of robots $r_i$ and $r_j$ depend on the symmetry of $\mathcal R(t)$ as follows:

 \textbf{{\bf Case 2.1 $\mathcal R(t)$ is free-path quasi regular:}} In this case, robots in $\mathcal R_g$ have direct movements and the robots in $\mathcal R_f$ do not move. If robots $r_i$ and $r_j$ move, they  move along $rad_i(t)$ and $rad_j(t)$ respectively. These two paths meet at $c_q$, the robot's destination point. Thus, robots do not collide during movements.

 \textbf{{\bf Case 2.2 $\mathcal R(t)$ is quasi regular but not free-path quasi regular:}} Robots in $\mathcal R$ have two types of movements: step-in and step-aside. The arguments are the same as in case 1. If $rad_i(t)$ and $rad_j(t)$ are coincident, then the paths of $r_i$ and $r_j$ are separated by $C_l(t)$. Otherwise, the paths are separated by $\mathcal L^*$.

 \textbf{{\bf Case 2.3 $\mathcal R(t)$ is not quasi regular:}} Both types of robots in $\mathcal R$ move in this case and have only step-in movements. If $rad_i(t)$ and $rad_j(t)$ coincide, then at most one of them moves. Without loss of generality, suppose $r_i$ moves. Robot $r_i$ moves to a point on $rad_i(t)$, and this point is at most as far away as the middle point of $\overline{r_i(t)r_j(t)}$ from $r_j(t)$. Thus, these two robots do not collide in this case. Next, consider the case when both robots move. In this case, $rad_i(t)$ and $rad_j(t)$ are different, and the paths of these two robots are separated by the bisector $\mathcal L^*$ of the angle $\angle{r_i(t)\mathcal O(t)r_j(t)}$.
\end{proof}

\begin{lemma}
\label{unqiue_SMP}
The multiplicity-creation phase creates a unique stable multiplicity point in finite time.
\end{lemma}

\begin{proof}
 We prove this lemma by considering each case separately:
\begin{itemize}
    \item   $\boldsymbol{\mathcal R(t)}\in$ \2:  In this case, only the robots having position in $\mathcal R_g(t)\cap S_{in}(t)$ move. The robots lying on $S(t)$ remain stationary, and hence the smallest enclosing circle $S(t)$ and O(t) remain invariant. Since $|\mathcal R_g|\ge6$ and $\mathcal S_{out}(t)\le4$, at least two robots from $\mathcal R_g$ lie inside $S(t)$.
    \begin{itemize}
        \item {\bf $\mathcal R(t)$ is free path quasi-regular:} In this case robots move straight towards $c_q$. Since robots move in straight lines towards $c_q$, the configuration remains free path quasi-regular until at least two robots reach $c_q$. Thus, at least two robots lie at $c_q$ within a finite time. Thus, at least two robots from $\mathcal R_g$ reach $c_q$ within finite time and create a stable multiplicity point at $c_q$. Since the robots in $\mathcal R_g$ have global weak multiplicity detection capability, every robot in $\mathcal R_g$ can identify $c_q$ as the multiplicity point and subsequently execute the gathering phase. This implies the lemma in this case.
        
        \item {\bf $\mathcal R(t)$ is not free path quasi-regular:} Robots in $\mathcal R_g(t)\cap S_{in}(t)$ move towards $\mathcal O(t)$. Since $S(t)$ remains invariant under the movements of these robots, $\mathcal O(t)$ also remains fixed. Consider a robot $r_i\in\mathcal R_g$ lying inside $S(t)$. If $(r_i(t,\mathcal O(t))$ does not contain any robot position, then robot $r_i$ moves straight towards $\mathcal O(t)$. Now, let $rad_j(t)$ be one of the neighbors of $rad_i(t)$ such that $r_j\in\mathcal R_g$ and it has a {\it side-aside movement} w.r.t. $\mathcal O(t)$. Then the paths of movements of $r_i$ and $r_j$ are separated by the bisector of the angle $\angle{r_i(t)\mathcal O(t)r_j(t)}$. Thus, robot $r_i$ reaches $\mathcal O(t)$ in finite time by moving along $\overline{r_i(t)\mathcal O(t)}$. Now suppose $(r_i(t),\mathcal O(t))$ has at least one robot position. In this, case, robot $r_i$ has a {\it step-aside movement} w.r.t. $\mathcal O(t)$. We show that robot $r_i$ gets a free corridor towards $\mathcal O(t)$ within finite time. Suppose robot $r_i$ lies on $C_k(t)$. During the {\it step-aside movement}, robot $r_i$ moves towards a point on the circle $C_k(t)$. The robot on $(r_i(t),\mathcal O(t))$, which is closest to $\mathcal O(t)$, does not move out of the line $rad_i(t)$. This implies that for each {\it step-aside movement} of robot $r_i$, at least one robot is removed from its straight line path to $\mathcal O(t)$. Since there is a finite number of robots on $(r_i(t),\mathcal O(t))$ and by \textbf{Lemma} \ref{lemma-mcp-1}, robots do not collide during this phase, within finite time, robot $r_i$ gets a robot-free straight path to $\mathcal O(t)$. Thus, within a finite time, at least two robots reach $\mathcal O(t)$ and make it a multiplicity point.
    \end{itemize}
     \textbf{Lemma} \ref{lemma-mcp-1} guarantees the uniqueness of the multiplicity point.

\item $\boldsymbol{\mathcal R(t)}\in$\3: We show that either a multiplicity point is created or the configuration is converted into a one configuration belonging to \2.
\begin{itemize}
\item {\bf $\boldsymbol{\mathcal R(t)}$ is free-path quasi-regular:} In this case only the robots in $\mathcal R_g$ move and they move towards $c_q$ following direct movements. Due to these movements the current robot configuration $\mathcal R(t)$ reaches to a robot configuration $\mathcal R(t^*)$ having any one of the following properties: (i) at least two robots reach $c_q$ and thus $\mathcal R(t^*)$ contains a multiplicity point or (ii) $\mathcal R(t^*)$ is free-path quasi-regular with $\mathcal R(t^*)\in$\2 or (iii) $\mathcal R(t^*)$ is free-path quasi-regular with $\mathcal R(t^*)\in$ \3. In cases (i) and (ii), we are done. Now consider the case (iii). Since $c_q$ remains invariant under the direct movements of the robots, $\mathcal R(t^*)$ has $c_q$ as its centre of quasi-regularity. This
implies that there are some robots in $\mathcal R_g$ whose distance from $c_q$ is reduced. Thus, the robot configuration satisfies property (i) or (ii) within a finite time. This implies the lemma in this case.

\item {\bf $\boldsymbol{\mathcal R(t)}$ is quasi-regular but not free-path quasi-regular:}   First consider the case when $c_q\neq\mathcal O(t)$. In this case, robots compute $\mathcal P'(t)$. The robots lying on $S(t)$ and not having positions in $\mathcal P'(t)$ take step-in
 movements towards $c_q$. By \textbf{Lemma} \ref{lemma-sec}, the points in the set $\mathcal P'(t)$ are sufficient to define $S(t)$. Since robots move towards $c_q$ along straight lines, the configuration remains quasi-regular but not free-path quasi-regular. Thus, within a finite time, configuration $\mathcal R(t)$ is converted to a configuration $\mathcal R(t^*)$ in \2. This implies the lemma in this case.

 Now consider the case when $c_q = \mathcal O(t)$. In this case, the robots are lying
 on $S(t)$ do not move. Thus, $S(t)$ does not change due to the robots' movements. The robots, lying inside $S(t)$ and not having free corridors to $\mathcal O(t)$, perform {\it step-aside movements}. These movements convert $\mathcal R(t)$ to a configuration $\mathcal R(t^*)$ having one of the following properties: (i) $\mathcal R(t^*)$ is free-path quasi-regular or (ii) $\mathcal R(t^*)$ is not quasi-regular or (iii) $\mathcal R(t^*)$ is quasi-regular, but not free-path quasi-regular, and the number of robots having free-paths is increased. We have discussed the case (i) above, and the case (ii) is discussed below. Consider the case (iii). If the centre of quasi-regularity of $\mathcal R(t^*)$ is different from $\mathcal O(t)$, then, as discussed above, we have the lemma in this case. Otherwise, the number of robots having free-paths is increased (since the smallest enclosing circle $S(t)$ remains the same and robots take the step-aside movements w.r.t. $\mathcal O(t)$). Since, by \textbf{Lemma} \ref{lemma-mcp-1}, robots do not collide during this phase, if this process is repeated, within finite time, we have either case (i) or case (ii). This completes the proof in this case.
\item {\bf $\boldsymbol{\mathcal{R}}(t)$ is not quasi regular:} In this case, robots compute $\mathcal P'(t)$. A robot $r_i$ moves if $r_i(t)\in S_{out}(t)$ and $r_i(t)\notin P'(t)$. By \textbf{Lemma} \ref{lemma-sec}, the robot positions in $\mathcal P'(t)$ are sufficient to maintain $S(t)$. The
moving robots only have step-in movements for $O(t)$. Due to these movements, $\mathcal R(t)$ is converted into configuration $\mathcal R(t^*)$ having any one of the following properties: (i) $\mathcal R(t^*)$ in \2 or (ii) $\mathcal R(t^*)$ in \3 with $|S_{in}(t^*)| > |S_{in}(t)|$. We are done in case (i). Now for case (ii), there are two possibilities: either $\mathcal R(t^*)$ is quasi-regular or $\mathcal R(t^*)$ is not quasi-regular. If $\mathcal R(t^*)$ is quasi-regular, then by the above case analysis, either we reach a configuration with a multiplicity point or a configuration which is not quasi-regular. Whenever a quasi-regular configuration $\mathcal R(t^*)$ is reached from another quasi-
 regular configuration $\mathcal R(t)$, we have $|S_{in}(t^*) |> |S_{in}(t)|$ (since the algorithm maintains the smallest enclosing circle in those cases). Since the number of robots is finite, within finite time either the system has a multiplicity point, or the configuration belongs to \2. The latter case also implies the former case in finite time. \textbf{Lemma} \ref{lemma-mcp-1} guarantees the uniqueness of the multiplicity point. Now we show that the multiplicity point created during this phase is stable. If at least two robots from $\mathcal R_g$ lie at the multiplicity point, then the multiplicity point is stable. Since robots are indistinguishable, the multiplicity point can contain a robot from $\mathcal R_f$. Since the robots' movements are collision-free during this phase, at most one point from $\mathcal R_f$ lies at the multiplicity point. If the multiplicity point contains a robot from $\mathcal R_f$, then
this robot can recognize that it lies at a multiplicity point using
its local weak multiplicity detection capability. Such a robot moves
from the multiplicity point only after the formation phase starts,
that is, when $|\mathcal R(t)| \leq N_f+1.$ At this time, at most one robot from $\mathcal R_g$ can lie outside
the multiplicity point. Hence, since $|\mathcal R_g|\geq 6$, at least
$|\mathcal R_g|-1\geq 5$ robots from $\mathcal R_g$ remain at the
multiplicity point. Therefore, the multiplicity point remains stable
even after the robot from $\mathcal R_f$ leaves it.
This completes the proof of the lemma.

\end{itemize}
\end{itemize}
\end{proof}

% \begin{lemma}
% \label{}
% Let $p_m$ be the stable multiplicity point created during the
% multiplicity-creation phase.

\begin{lemma}
\label{lemma-gf-collision-free}
If the multiplicity point created during the multiplicity creation
phase does not contain a robot from $\mathcal R_f$, then
\textsc{MoveToDestination}() provides collision-free movements during
the gathering and formation phases.
\end{lemma}

\begin{proof}
Since the multiplicity point $p_m$ does not contain a robot from
$\mathcal R_f$, the formation phase cannot start before all the robots
in $\mathcal R_g$ gather at $p_m$. Indeed, as long as at least one
robot from $\mathcal R_g$ remains outside $p_m$, the number of
distinct occupied positions is greater than $N_f+1$. Hence, the
condition $|\mathcal R(t)|=N_f+1$ required for the start of the
formation phase is not satisfied. Thus, in this case, the gathering phase and the formation phase are
executed sequentially.

\textbf{During the gathering phase:}
Since robots in $\mathcal R_f$ do not have global weak multiplicity
detection capability, they continue acting according to the
multiplicity creation phase until $|\mathcal R(t)|=N_f+1$. Their
movements are limited to step-in and step-aside movements, both with
respect to a single point, either $\mathcal O(t)$ or $c_q$.
The robots in $\mathcal R_g$ move toward $p_m$ by either direct or
step-aside movements. Let $r_i$ and $r_j$ be two active robots in round $t$.

\begin{itemize}
    \item \textbf{Robots $r_i,r_j\in \mathcal R_g$:}
    If the line segments $\overline{r_i(t)p_m}$ and
    $\overline{r_j(t)p_m}$ are not coincident, then the paths of
    movements of $r_i$ and $r_j$ are completely separated by the
    bisector of the angle $\angle r_i(t)p_mr_j(t)$.
    Next suppose that $\overline{r_i(t)p_m}$ and
    $\overline{r_j(t)p_m}$ are coincident. Without loss of generality,
    suppose that $r_i(t)$ lies on $(r_j(t),p_m)$. Let $r_i(t)$ lie on
    $C_l^*(t)$ and $r_j(t)$ lie on $C_k^*(t)$ for some $l,k$. Let
    $C_{kl}^*(t)$ be the circle having centre at $p_m$ and lying
    midway between $C_l^*(t)$ and $C_k^*(t)$. Then the movement paths
    of $r_i$ and $r_j$ are completely separated by the circle
    $C_{kl}^*(t)$.

    \item \textbf{Robots $r_i,r_j\in \mathcal R_f$:}
    Robots $r_i$ and $r_j$ continue to act according to the
    multiplicity creation phase. Hence, by Lemma~\ref{lemma-mcp-1},
    their movements are collision-free.

    \item \textbf{Robots $r_i,r_j$ belong to different groups:}
    Without loss of generality, suppose that $r_i\in\mathcal R_g$ and
    $r_j\in\mathcal R_f$.
    If $p_m=\mathcal O(t)$, then both robots act exactly as in the
    multiplicity creation phase. Hence, by Lemma~\ref{lemma-mcp-1},
    they do not collide. Now consider the case when $p_m\neq\mathcal O(t)$.
    If $r_j$ does not move, then $r_i$ does not collide with $r_j$.
    If $r_j$ has a step-aside movement during the multiplicity
    creation phase, then according to the gathering rule, $r_i$ does
    not move. Hence, no collision occurs.
    Finally, if both robots move, then $r_j$ has a step-in movement
    with respect to $\mathcal O(t)$, while $r_i$ has either a direct
    movement or a step-aside movement with respect to $p_m$. If $r_i$ has a direct movement and the segments
    $\overline{r_i(t)p_m}$ and $\overline{r_j(t)\mathcal O(t)}$
    intersect at a point $x$, then by construction the destination of
    $r_j$ lies on the open segment $(r_j(t),x)$. Hence, the two robots
    do not collide. If $r_i$ has a step-aside movement and $r_i(t)$ lies on
    $C_l^*(t)$, then its destination lies on $C_l^*(t)\cap V_o$.
    If $\overline{r_j(t)\mathcal O(t)}$ does not intersect
    $C_l^*(t)\cap V_o$, then the paths do not intersect.
    Otherwise, let $x$ be the intersection point. Then the paths are
    separated by the bisector of the segment $\overline{r_i(t)x}$.
\end{itemize}
Thus, the gathering phase is collision-free.

\textbf{During the formation phase:}
Once all robots in $\mathcal R_g$ gather at $p_m$, we have
$|\mathcal R(t)|=N_f+1$, and the robots in $\mathcal R_f$ start the
formation phase.
During this phase only the robots in $\mathcal R_f$ move, and they
have either direct movement or step-out movement with respect to
$\mathcal O(t)$. Their movements are ordered according to their
distance from $\mathcal O(t)$.
Consider two robots $r_i,r_j\in\mathcal R_f$ active in round $t$.
Let $rad_i(t)$ and $rad_j(t)$ intersect $S(t)$ at $u_i$ and $u_j$,
respectively.
If $u_i\neq u_j$ and both $u_i$ and $u_j$ do not contain robot
positions, then the paths of the direct movements of $r_i$ and $r_j$
are separated by the bisector of the angle
$\angle r_i(t)\mathcal O(t)r_j(t)$.
If $u_i\neq u_j$ and exactly one of them contains a robot position,
say $u_i$, then the destination point of $r_i$ is computed so that it
lies on a different side of the bisector of
$\angle r_i(t)\mathcal O(t)r_j(t)$ from $u_j$. Hence, $r_i$ and $r_j$
do not collide.
Finally, if $u_i=u_j$, then the robots lie on the same radial line.
Without loss of generality, suppose that $r_j$ lies on $C_{k-1}(t)$
and $r_i$ lies on $C_{k-2}(t)$. Then the destination point of $r_i$
lies on $C_{k-1}(t)$, and the movement paths of $r_i$ and $r_j$ are
completely separated by $C_{k-1}(t)$.
Therefore, the formation phase is collision-free. This completes the
proof.
\end{proof}

\begin{lemma}
\label{lemma-overlap-transient}
If the multiplicity point created during the multiplicity creation
phase contains a robot from $\mathcal R_f$, then the gathering and
formation phases may overlap. In that case, no two robots from
$\mathcal R_f$ collide during the formation phase. Moreover, if a
robot from $\mathcal R_g$ collides with a robot from $\mathcal R_f$,
the resulting additional multiplicity point is unstable and disappears
once the robot from $\mathcal R_f$ moves away from it.
\end{lemma}

\begin{proof}
Suppose that the multiplicity point $p_m$ created during the
multiplicity creation phase contains a robot from $\mathcal R_f$.
Since the movements during the multiplicity creation phase are
collision-free, at most one robot from $\mathcal R_f$ can lie at
$p_m$.
The gathering and formation phases may overlap only in the following
situation: exactly one robot from $\mathcal R_g$ remains outside
$p_m$, while at least one robot from $\mathcal R_f$ not lying at
$p_m$ finds $|\mathcal R(t)|=N_f+1$ and starts the formation phase.
This is possible because one robot from $\mathcal R_f$ already lies
at $p_m$ and the robots in $\mathcal R_f$ have only local weak
multiplicity detection capability.
Let $r_i\in\mathcal R_g$ be the unique robot not lying at $p_m$, and
let $r_j\in\mathcal R_f$ be a robot active in the same round. Robot
$r_i$ has either a direct or a step-aside movement toward $p_m$,
whereas $r_j$ has a step-out movement toward the boundary of $S(t)$.
Hence, the paths of these two robots may intersect. Due to non-rigid
movement, the adversary may stop both robots at such an intersection
point, thereby creating another multiplicity point.
However, this multiplicity point is not stable, because it contains
only one robot from $\mathcal R_g$. Once the robot from
$\mathcal R_f$ moves away from that point, the multiplicity is
broken.
Moreover, the collision-avoidance argument used in the formation
phase depends only on the movements of robots in $\mathcal R_f$.
Hence, no two robots from $\mathcal R_f$ collide during the
formation phase, even when the overlap occurs.
Therefore, during the overlap of the gathering and formation phases,
any additional multiplicity point created by a collision between a
robot from $\mathcal R_g$ and a robot from $\mathcal R_f$ is only
temporary and cannot become stable. This completes the proof.
\end{proof}

\begin{lemma}
\label{GPG}
The gathering phase gathers all robots in $\mathcal R_g$ at the unique stable
multiplicity point.
\end{lemma}

\begin{proof}
 By \textbf{Lemma} \ref{lemma-mcp-1}, the multiplicity point created during the multiplicity creation phase is stable. Let $p_m$ denote the multiplicity point created during the multiplicity creation phase. Robots in $\mathcal R_g$ can identify the multiplicity point using their global weak multiplicity detection capability. First, consider the case when $p_m$ contains no robot from $\mathcal R_f$. The robots in $\mathcal R_f$ do not have global weak multiplicity detection capability. Thus, they can not detect the start of the gathering phase and continue performing the actions according to the multiplicity creation phase. By \textbf{Lemma} \ref{lemma-gf-collision-free}, robots do not collide during movements in this case. Thus, $p_m$ remains as the unique multiplicity point in the system. Consider a robot $r_i\in\mathcal R_g$. If $r_i(t)=p_m$, robot $r_i$ does not move. If $(r_i(t),\mathcal O(t))$ does not contain any other robot positions, then robot $r_i$ has direct movement towards $\mathcal O(t)$. Since robots have collision-free movements, robot $r_i$ reaches $\mathcal O(t)$ in finite time. Finally, suppose that $(r_i(t),\mathcal O(t))$ contains at least one robot position. We must show that robot $r_i$ gets a free path to $\mathcal O(t)$ in finite time. Robot $r_i$ adopts {\it step-aside movement} w.r.t. $p_m$. Now, if $r_i$ and all the robots on $(r_i(t),p_m)$ move in each round, then $r_i$ may not get a free path to $p_m$ (all these robots may land on a new line passing through $p_m$). To tackle this, we consider the possible scenarios when all the robots on $(r_i(t),p_m)$ move. If $(r_i(t),p_m)$ contains all robots from $\mathcal R_f$ and $\mathcal R(t)$ is quasi-regular but no free-path quasi regular with $c_q=\mathcal O(t)$, then the robots on $(r_i(t),p_m)$ perform {\it step-aside movements}. According to our algorithm, robot $r_i$ does not move in this case. Due to movements of the robots on $(r_i(t),\mathcal O(t))$, robot $r_i$ either gets a free path to $p_m$ or the characteristic of the current configuration is changed. In the first case, we are done. In the second case, robot $r_i$ makes a {\it step-aside} movement and the robots on $(r_i(t),p_m)$ do not move. Thus, $r_i$ gets a free path to $p_m$. Now, if $(r_i(t),\mathcal O(t))$ also contains robots from $\mathcal R_g$, then each movement of the robots from $(r_i(t),p_m)$ reduces the number of robots on the path of $r_i$ to $p_m$. Since the number of robots is finite, robot $r_i$ gets a free direct path to $p_m$ within finite time. Now consider the other scenarios when the $\mathcal R_f$ robots do not have step-aside movements. The $\mathcal R_f$ robots can have step-in movements w.r.t. $\mathcal O(t)$ or $c_q$ depending on the configuration. If robots in $\mathcal R_f$ do not move, then the movements of the robots on $(r_i(t),p_m)$ reduce the number of robots on the direct path of $r_i$ to $p_m$ and in finite time $r_i$ gets a free path to $p_m$. Consider the scenario when robots in $\mathcal R_f$ have step-in movements. Only the robots on $S(t)$ can have step-in movements. We show that robots having step-in movements do not become collinear with $r_i$ and $p_m$. Let $r_j\in\mathcal R_f$ be a robot which takes a step-in movement during this phase. If $\overline{r_i(t)p_m}$ does not intersect $\overline{r_j(t)d}$, then we have nothing to prove where $d\in\{\mathcal O(t),c_q\}$. Otherwise, let $x$ be the intersection point between $\overline{r_i(t)p_m}$ and $\overline{r_j(t)p}$. The destination point of $r_j$ during its step-in movement lies on $(r_j(t),x)$. The destination point of $r_i$ lies in $V_o$. Thus, the paths of these two robots are separated by the perpendicular line to $\overline{r_i(t)r_j(t)}$ passing through the point $x$. This completes the proof in this case.

 Consider when a robot from $\mathcal R_f$ lies at $p_m$. In this case, the gathering phase overlaps with the formation phase. In the proof of \textbf{Lemma} \ref{lemma-gf-collision-free}, we have seen that no two robots $\mathcal R_f$ collided during the formation phase. However, we have shown that exactly one robot from one $\mathcal R_g$, say $r_i$, may collide with a robot from $\mathcal R_f$. This will create a multiplicity point, say $p^*$. According to the algorithm, robots from $\mathcal R_g$ do not move when they find that they lie at a multiplicity point. Since all but one robot from $\mathcal R_g$ lies at $p_m$, a stable multiplicity point, no other robot from $\mathcal R_g$ moves to $p^*$. Thus, the multiplicity point $p^*$ is unstable: once the robot from $\mathcal R_f$ lying at $p^*$ moves, this multiplicity point is broken. The formation phase starts when $\mathcal R(t)=|N_f|$. Thus, according to the algorithm,  the robots from $\mathcal R_f$ lying at the multiplicity point move during this phase. This implies that the multiplicity of $p^*$ will be broken in finite time. Once this happens, robot $r_i$ moves towards $p_m$. Robot $r_i$ may again collide with some other robot from $\mathcal R_f$; however, its distance from $p_m$ is decreased by at least $\delta$ in each movement. Thus, within a finite time, it reaches $p_m$. This completes the proof of the lemma.  
 \end{proof}

 \begin{lemma}
\label{Fpp}
The formation phase places all robots in $\mathcal R_f$ at distinct positions on the
boundary of a common circle.
\end{lemma}

\begin{proof}
This phase starts when robots in $\mathcal R_f$ find $|\mathcal R(t)|\le N_f+1$. We prove the lemma by considering the following two cases separately: (i) when $p_m$ does not contain a robot from $\mathcal R_f$ and (ii) when $p_m$ contains a robot from $\mathcal R_f$.

 \begin{itemize}
     \item {\bf $p_m$ does not contain a robot from $\mathcal R_f$:} In this case, only the robots in $\mathcal R_f$ moves during this phase. By \textbf{Lemma} \ref{lemma-gf-collision-free}, robots do not collide with each other during movements in this case. The robots lying on $S(t)$ do not move. The robots from $\mathcal R_f$ lying inside or at the multiplicity point move according to {\it step-out movement}. Thus, $S(t)$ remains invariant during this phase. The movements of the robots are ordered: a robot $r_i\in \mathcal R_f$ moves if it lies on $C_{k-1}(t)$ or on $C_{k-2}(t)$, where $C_k(t)=S(t)$. Robots lying on $C_{k-2}(t)$ moves to $C_{k-1(t)}$ when there is exactly one robot position on $C_{k-1}(t)$. A robot lying on $C_{k-1}(t)$ moves to $S(t)$. Suppose $r_i$ lies on $C_{k-1}(t)$ and $x$ is the intersection point between $rad_i(t)$ and $S(t)$. If $x$ contains no robot position, then $r_i$ moves straight to $x$. Otherwise, it computes a point on $S(t)$ which does not contain any robot position. By \textbf{Lemma} \ref{lemma-gf-collision-free}, no two robots compute the same destination point on $S(t)$ (otherwise, a collision could happen). Since $S(t)$ remains invariant under the robots' movements, robot $r_i$ reaches $S(t)$ within finite time. If $r_i$ lies on $C_{k-2}(t)$, then by the same approach discuss above, robot $r_i$ first reaches $C_{k-1}(t)$ and then from there to $S(t)$. This implies the lemma in this case.  
     
    \item {\bf $p_m$ contains a robot from $\mathcal R_f$:} Let $r_j\in\mathcal R_f$ lie  at $p_m$. In this case, the formation phase may start before the end of the gathering phase. It may start before exactly one robot, say $r_k$, reaches $p_m$. Thus, the circle $S(t)$ may change with the robot's movement $r_k$. However, the robot $r_k$ reaches $p_m$ within a finite time. Once $r_k$ reaches $p_m$, the smallest enclosing circle of the robot positions becomes stable. Thus, we can conclude the proof in this case using the same argument as in the above.
 \end{itemize}
 \label{pag}
\end{proof}

From the above lemmas, we obtain our main result, stated in Theorem~\ref{thm:main}.

\begin{theorem}
\label{thm:main}
Let
\(
\mathcal R=\mathcal R_g\cup \mathcal R_f
\)
be a set of anonymous, oblivious, fully disoriented robots operating under the
$\mathrm{SSYNC}$ scheduler with non-rigid movements. Suppose that
\(
|\mathcal R_g|\ge6
\)
and
\(
|\mathcal R_f|\ge2
\), and that initially all robots occupy distinct positions.
If robots in $\mathcal R_g$ have global weak multiplicity detection and robots in
$\mathcal R_f$ have local weak multiplicity detection together with the knowledge of
\(
|\mathcal R_f|,
\)
then Algorithm \textsc{PatternFormation()} terminates in finite time such that
all robots in $\mathcal R_g$ gather at one point and all robots in $\mathcal R_f$ occupy
distinct positions on a common circle.
\end{theorem}

\section{Conclusions} 
\label{con}
%\vspace*{-0.2cm}
In this paper, we have extended the study initiated in~\cite{Conflict-1}. We propose a distributed algorithm that solves the problem for disoriented semi-synchronous robots with non-rigid movements. Our algorithm works without any kind of assumption on the local coordinate axes of the robots. The proposed algorithm assumes global weak multiplicity detection capability for the robots, which solves the gathering problem and local weak multiplicity detection capability for the robots, which solves the circle formation problem. The proposed algorithm works with robots that do not have rigid movements. One of the future directions of this work is to extend it to disoriented asynchronous robots.  It would also be interesting to remove the assumptions $|\mathcal{R}_g|\ge 6$ and the knowledge of $|\mathcal{R}_f|$.

% \section*{Acknowledgements}
% ...

\section*{Funding}
Animesh Maiti and Prakhar Shukla were supported by the INSPIRE Fellowship of the Department of Science and Technology (DST), Government of India.

\section*{Declaration of competing interest}
The authors declare that they have no known competing financial
interests or personal relationships that could have appeared to
influence the work reported in this paper.


\bibliographystyle{plain} 
\bibliography{Biblo}

\end{document}